\documentclass[11pt,english]{article}
\usepackage{color}
\usepackage[T1]{fontenc}
\usepackage[latin9]{inputenc}
\usepackage[a4paper]{geometry}
\usepackage{textcomp}
\usepackage{amsthm}
\usepackage{amsmath}
\usepackage{amssymb}
\usepackage{esint}
\usepackage{bbm}
\usepackage{hyperref}
\usepackage{babel}
\usepackage{amscd}
\usepackage{amsfonts}
\usepackage{amsthm}
\usepackage{mathrsfs}
\usepackage{dsfont}
\usepackage{mathtools}
\usepackage{enumerate}
\usepackage{bm}
\usepackage{bbm}
\usepackage{verbatim}
\usepackage{enumitem}
\newlist{abbrv}{itemize}{1}
\setlist[abbrv,1]{label=,labelwidth=1.2in,align=parleft,itemsep=0.1\baselineskip,leftmargin=!}

\makeatletter

\DeclareFontEncoding{LGR}{}{}
\DeclareTextSymbol{\~}{LGR}{126}
\DeclareMathOperator{\SPAN}{span}

\newcommand{\vph}{\widehat{\Phi}_{\Lambda}}
\newcommand{\vphc}{\Phi_{\Lambda}}
\newcommand{\tr}{\mathrm{Tr}}
\newcommand{\abs}[1]{\left| #1 \right|}

\newcommand{\el}{\tilde{\varphi}}

\newcommand{\scp}[2]{\big\langle #1 , #2 \big\rangle}
\newcommand{\SCP}[2]{\big\langle #1 , #2 \big\rangle}
\newcommand{\bra}[1]{\langle #1 |}
\newcommand{\ket}[1]{| #1 \rangle}
\newcommand{\ketbra}[2]{| #1 \rangle \langle #2 |}
\newcommand{\ketbr}[1]{| #1 \rangle \langle #1 |}
\newcommand{\norm}[1]{\left\| #1 \right\|}
\newcommand{\Hilbert}{\mathscr{H}}

\renewcommand{\Im}{\mathrm{Im}}
\newcommand{\RRR}{\mathbb{R}}
\newcommand{\NNN}{\mathbb{N}}
\newcommand{\id}{\mathbbm{1}}
\newcommand{\HS}{\mathrm{HS}}
\newcommand{\Term}{\mathrm{Term}}
\newcommand{\as}{\mathrm{as}}
\newcommand{\mbp}{P}
\newcommand{\mbq}{Q}

\newtheorem{theorem}{Theorem}[section]
\newtheorem{lemma}[theorem]{Lemma}
\newtheorem{notation}[theorem]{Notation}
\newtheorem{remark}[theorem]{Remark}
\newtheorem{definition}[theorem]{Definition}
\newtheorem{assumption}[theorem]{Assumption}
\newtheorem{proposition}[theorem]{Proposition}

\allowdisplaybreaks[1]

\begin{document}
\title{Mean-field Dynamics for the Nelson Model with Fermions}

\author{
Nikolai Leopold\footnote{IST Austria (Institute of Science and Technology Austria), Am Campus 1, 3400 Klosterneuburg, Austria. E-mail: {\tt nikolai.leopold@ist.ac.at}} \ and 
S\"oren Petrat\footnote{Jacobs University, Department of Mathematics and Logistics, Campus Ring 1, 28759 Bremen, Germany. E-mail: {\tt s.petrat@jacobs-university.de}}}

\date{July 8, 2019}

\maketitle

\begin{abstract}
\noindent
We consider the Nelson model with ultraviolet cutoff, which describes the interaction between non-relativistic particles and a positive or zero mass quantized scalar field. We take the non-relativistic particles to obey Fermi statistics and discuss the time evolution in a mean-field limit of many fermions. In this case, the limit is known to be also a semiclassical limit. We prove convergence in terms of reduced density matrices of the many-body state to a tensor product of a Slater determinant with semiclassical structure and a coherent state, which evolve according to a fermionic version of the Schr\"odinger-Klein-Gordon equations.
\end{abstract}

\noindent
\textbf{MSC class:} 35Q55, 81Q05, 81T10, 82C10

\section{Introduction}
\label{section: Nelson Introduction}
Interacting many-body systems are very difficult to analyze, and analytic or numerical solutions are usually not feasible. Therefore simpler effective equations are used to analyze these systems throughout the sciences. These approximations work very well in many settings and can be derived with heuristic arguments and good intuition. In mathematical physics the question of a rigorous justification of such effective equations is an active field of research, starting in the 1970's with works such as \cite{hepp,lanford,braunhepp,ginibrevelo,spohn:1981} (see \cite{spohneffbook} for an excellent overview). Sparked by the 2001 nobel prize for the experimental realization of a Bose-Einstein condensate there has been great interest in the derivation of effective equations for bosonic systems (we refer to \cite{lieb:2005,schleinbook} for references and an overview of the topic). More recently, there was an increasing interest in the evolution of many fermion systems. This started already in the 80's with the works \cite{narnhofer:1981,spohn:1981}, which introduce the mean-field limit for fermions and prove convergence to the classical Vlasov equation. Convergence to the Hartree-Fock equations was proven in 2004 in \cite{erdoes:2004}, where the authors consider short times and analytic interaction potentials, and in particular highlight the importance of the semiclassical structure in the derivation. The generalization to arbitrary times and a larger class of bounded interaction potentials was achieved in \cite{benedikter_2013,benedikter:2014}, see also \cite{petrat_2015} for a slightly different proof. The more recent work \cite{porta:2016} extended the results to Coulomb interaction (see also \cite{saffirio_proc} for weaker singularities), assuming a property of the Hartree-Fock dynamics that the authors only prove for the special situation of translation invariant initial data. The article \cite{benedikter:2015_2} covers mixed initial states. Several other results for different time scales (without semiclassical structure) were obtained in \cite{bardos:2003,bardos:2004,bardos:2007} for a coupling constant $N^{-1}$, in \cite{froehlich:2011} for a coupling constant $N^{-1}$ and Coulomb interaction, and in \cite{petrat_2015,bach:2015,petrat:2016} for a coupling constant $N^{-2/3}$ and singular interactions potentials. In particular, in \cite{petrat:2016} convergence to the fermionic Hartree equations is proven for Coulomb interaction with a convergence rate that distinguishes the mean-field equation from the free equation. Let us also mention the article \cite{benedikter_2018}, where the authors discuss the Bogoliubov-de Gennes equations for fermions, which is an approximation more precise than Hartree-Fock theory. In particular, they derive these equations assuming that the states are quasifree for all times. These works show that many aspects of the mean-field regime of weakly correlated bosons and fermions that interact via a pair potential are well understood by now. However, less attention has been paid to systems in which the interaction between the particles is mediated by a second quantized radiation field. Also here effective equations are of great importance because quantized radiation fields are described on Fock space, i.e., a Hilbert space for an arbitrary number of particles. The complexity of such systems is reduced tremendously when the quantized field is approximated by a pair potential or a classical radiation field. 
The articles \cite{davies_1979, hiroshima_1998, teufel2} show that the quantized radiation field can sometimes be replaced by a two-particle interaction if the particles are much slower than the bosons of the radiation field. Moreover it is possible to derive classical field equations from second quantized models \cite{ginibrenironivelo, falconi, ammarifalconi, frankenschlein, frankgang, leopold, leopold3, griesemer, correggi_2017, correggi_2018}. 
While these works focus on bosonic systems or systems with a small number of fermions, the present paper seems to be the first that considers a many particle limit of fermions which interact by means of a quantized radiation field. The scaling, which will be explained in the following, can been seen as a fermionic mean-field limit because it is chosen such that the source term of the radiation field can effectively be replaced by its mean value. Moreover, it can be viewed as a second quantized analogue of the fermionic mean-field model of \cite{benedikter_2013}. 

We consider $N$ identical fermions that interact by means of a quantized scalar field.
The state of the radiation field is represented by elements of the bosonic Fock space $\mathcal{F}_s \coloneqq \bigoplus_{n \geq 0} L^2(\mathbb{R}^3)^{\otimes_s n}$, where the subscript $s$ indicates symmetry under interchange of variables. The Hilbert space of the whole system is 
\begin{align}
\mathcal{H}^{(N)} \coloneqq L^2_{\as}\left( \mathbb{R}^{3N} \right) \otimes \mathcal{F}_s.
\end{align}
Here the subscript ``$\as$'' indicates antisymmetry under exchange of variables. An element $\Psi_N \in \mathcal{H}^{(N)}$ is a vector $\big( \Psi_N^{(n)} \big)_{n \in \mathbb{N}_0}$ with $\Psi_N^{(n)} \in L^2_{\as}(\RRR^{3N}) \otimes L^2_s(\RRR^{3n})$ and
\begin{align}
\norm{\Psi_N}^2 = \sum_{n=0}^{\infty} \int d^{3N}x \, d^{3n}k \, \abs{\Psi_{N}^{(n)}(X_N,K_n)}^2 < \infty,
\end{align}
where we use the short-hand notation $X_N = (x_1, \ldots, x_N)$ and $K_n = (k_1, \ldots k_n)$. We define the annihilation and creation operators by
\begin{align}\label{eq: Nelson pointwise creation and annihilation operators}
\begin{split}
\left( a(k) \Psi_N \right)^{(n)} (X_N,k_1,\ldots,k_n) &= ( n + 1)^{1/2} \Psi_N^{(n+1)}(X_N,k,k_1, \ldots, k_n), \\
\left( a^*(k) \Psi_N \right)^{(n)} (X_N,k_1,\ldots,k_n) &=
n^{-1/2} \sum_{j=1}^n \delta(k- k_j) \Psi_N^{(n-1)}(X_N, k_1, \ldots, \hat{k}_j, \ldots, k_n),
\end{split}
\end{align}
where $\hat{k}_j$ means that $k_j$ is left out in the argument of the function. They satisfy the commutation relations
\begin{align}
\label{eq: Nelson canonical commutation relation}
[a(k), a^*(l) ] &=  \delta(k-l), \quad
[a(k), a(l) ] = 
[a^*(k), a^*(l) ] = 0.
\end{align}
We choose units such that $\hbar=1=c$. The dispersion relation is then given by $\omega(k) = ( \abs{k}^2 + m^2 )^{1/2}$ with mass $m \geq 0$. We define the form factor of the radiation field by
\begin{equation}
\label{eq: Nelson cut off function}
\tilde{\eta}(k) = \frac{(2 \pi)^{-3/2}}{\sqrt{2 \omega(k)}} \id_{\abs{k}\leq \Lambda}(k), \quad
\text{with} \;  \id_{\abs{k}\leq \Lambda}(k) =
\begin{cases} 
1 &\text{if } \abs{k} \leq \Lambda , \\
0 &\text{otherwise} .
\end{cases}
\end{equation}
Here, $\Lambda$ is a momentum cutoff and we assume $\Lambda \geq 1$. The field operator is given by
\begin{align}
\vph(x) =  \int d^3k \, \tilde{\eta}(k)
\left( e^{ikx} a(k) + e^{-ikx} a^*(k)  \right) 
\end{align}
and the free Hamiltonian of the scalar field is the self-adjoint operator 
\begin{align}
H_f &= \int d^3k \, \omega(k) a^*(k) a(k)
\end{align}
with
\begin{align}
\label{eq: Nelson field energie domain}
\mathcal{D}(H_f) = \bigg\{ \Psi_N \in \mathcal{H}^{(N)}: \sum_{n=1}^{\infty} \int d^{3N}x \, d^{3n}k \, \Bigg\lvert\sum_{j=1}^n \omega(k_j)  \Psi_N^{(n)}(X_N, K_n)\Bigg\rvert^2  < \infty \bigg\}.
\end{align}
The full system is described by the Nelson Hamiltonian 
\begin{align}
\label{eq: Nelson Hamiltonian}
H_N =& \sum_{j=1}^N  \left(  - \Delta_j  +  \vph(x_j) \right) +  \delta_N H_f.
\end{align}
The factor $\delta_N$ is an arbitrary particle number dependent scaling parameter that allows to scale the field energy. The Nelson Hamiltonian is self-adjoint on the domain $\mathcal{D}\left(H_N \right) = \big(H^2(\RRR^{3N})\otimes \mathcal{F}_s\big) \cap \mathcal{D}(H_f)$, where $H^2$ denotes the second Sobolev space. This can be shown by applying Kato's theorem as in \cite{nelson, spohn}. The time evolution of the wave function $\Psi_{N,t}$ is governed by the Schr\"odinger equation
\begin{align}\label{eq:Schroedinger_equation_microscopic}
i \partial_t \Psi_{N,t} = N^{-1/3}  H_N \Psi_{N,t}.
\end{align}
The appearance of $N^{-1/3}$ in \eqref{eq:Schroedinger_equation_microscopic}  stems from the fact that we are interested in initial conditions which are localized in a volume of order one. Then, due to the Fermi statistics, the average kinetic energy per fermion is of order $N^{2/3}$, and the average momentum per fermion of order $N^{1/3}$. Therefore, we rescale time so we track the particles only while they move in the volume of order one, i.e., we go to time scales $N^{-1/3}$. 
This gives rise to a factor  $N^{1/3}$ in front of the time derivative.

If we use the Schr\"odinger equation \eqref{eq:Schroedinger_equation_microscopic} to compute the Ehrenfest equation for the field operator, we obtain
\begin{align}
\label{eq: Ehrenfest equations}
&\big[ \partial_t^2 + N^{-2/3} \delta_N^2 (- \Delta_x + m^2)  \big] \SCP{\Psi_{N,t}}{\vph(x) \Psi_{N,t}}
\nonumber \\
&\qquad = - N^{1/3} \delta_N (2 \pi)^{-3} \int d^3k \, e^{-ikx} \id_{\abs{k}\leq \Lambda}(k) \frac{1}{N} \SCP{\Psi_{N,t}}{\sum_{j=1}^N e^{ik x_j} \Psi_{N,t}},
\end{align}
where $\scp{\cdot}{\cdot}$ is the scalar product on $\mathcal{H}^{(N)}$, and the $x_j$'s on the right-hand side refer to the variables in $L^2_{\as}\left( \mathbb{R}^{3N} \right)$ that are integrated. Note that the integral on the right-hand side is proportional to $N^{-1}$ times the smeared out electron density (i.e., for $\Lambda\to\infty$ the electron density). Thus, for our initial conditions, the integral is a function of order one in a volume of order one. Equation~\eqref{eq: Ehrenfest equations} also shows that not only the coupling constant in front of the radiation field (which we set equal to one) but also $\delta_N$ determines the variation of the mean of the field operator.
While our main result Theorem \ref{theorem: main theorem} holds for arbitrary $\delta_N$, we believe that two choices are of particular interest.
\begin{enumerate}
\item 
For $\delta_N = N^{1/3}$ the velocities of the electrons and the bosons scale equally. Moreover, it ensures that the right-hand side of \eqref{eq: Ehrenfest equations}  and hence the variation of the mean of the field operator is of order $N^{2/3}$. This gives rise to the interesting effective evolution equations \eqref{effective_eqs} which capture the effect of the interaction.
\item
If we set $\delta_N =1$ our model corresponds to an unscaled system whose dynamics is studied for time scales of order $N^{-1/3}$. This is interesting because usually mean-field results for systems with two-particle interaction require a scaling of the coupling constant. It should be noted that most of the electrons travel on a distance of order one and hence could interact with the other electrons. However, a look at  \eqref{eq: Ehrenfest equations} shows that the group velocity of the bosons is too slow to mediate an interaction between the electrons. This implies (see Theorem~\ref{theorem: free evolution}) that the electrons effectively evolve like free particles in an external potential.
\end{enumerate}

Further insight concerning the scaling can be gained if we set $\varepsilon_N = N^{-1/3}$ and multiply \eqref{eq:Schroedinger_equation_microscopic} by $\varepsilon_N$. This gives 
\begin{align}
\label{eq:Schroedinger_equation_microscopic2}
i  \varepsilon_N \partial_t \Psi_{N,t} = \Bigg[  \sum_{j=1}^N \Big(  - \varepsilon_N^2 \Delta_j + N^{-1/2} \varepsilon_N^{1/2} \vph(x_j) \Big) + \varepsilon_N N^{-1/3} \delta_N H_f   \Bigg] \Psi_{N,t}.
\end{align}
Here, the factor $\varepsilon_N$ appears exactly where the physical constant $\hbar$ appears in the Schr\"odinger equation. Thus, for $\delta_N = N^{1/3}$, our limit can be viewed as a combined weak coupling (the $N^{-1/2}$ in front of the interaction term) and semiclassical limit. Moreover, it displays a connection to the fermionic mean-field scaling considered in \cite{benedikter_2013}, i.e., to the model
\begin{align}\label{eq:BPS_model}
i \varepsilon_N \partial_t \chi_{N,t}
&= \Bigg[ - \sum_{j=1}^N \varepsilon_N^2 \Delta_j + \frac{1}{N} \sum_{i<j}^N V(x_i - x_j)  \Bigg] \chi_{N,t}
\end{align}
with $\chi_{N,t} \in L^2_{\as}(\mathbb{R}^{3N})$, and some $V:\RRR^3\to\RRR$. Like in \cite{benedikter_2013} it will be crucial for us to consider initial data with a semiclassical structure, meaning that the kernel of the one-particle reduced density matrix is concentrated along its diagonal (see Remark~\ref{remark_sc_struc} for more details). 

We assume the initial states to be approximately of product form
\begin{align}
\label{eq: Nelson initial product state}
\bigwedge_{j=1}^N \varphi_j^0 \otimes W( N^{2/3} \alpha^0) \Omega.
\end{align}
Here, $\alpha^0 \in L^2(\RRR^3)$, $\bigwedge_{j=1}^N \varphi_j^0$ denotes the antisymmetrized tensor product (wedge product) of orthonormal $\varphi_1^0,\ldots,\varphi_N^0 \in L^2(\RRR^3)$, $\Omega$ denotes the vacuum in $\mathcal{F}_s$ and $W$ is the Weyl operator
\begin{align}
\label{eq: Nelson Weyl operator}
W(f) \coloneqq \exp \left(  \int d^3k \, \Big( f(k) a^*(k) - \overline{f(k)} a(k) \Big) \right)
\end{align} 
for all $f \in L^2(\mathbb{R}^3)$ ($\overline{f(k)}$ denotes the complex conjugate of $f(k)$). In such a state the only correlations are due to the antisymmetry of the electron wave function. During the time evolution correlations emerge but the product structure (as will be shown) is preserved in the limit $N \rightarrow \infty$ on the level of reduced density matrices. This suggests to approximate the action of the scaled field operator $N^{-2/3} \vph$ on $\Psi_{N,t}$ by a classical radiation field $\vphc(x,t)$ and replace the right-hand side of \eqref{eq: Ehrenfest equations} by a coupling to the mean electron density. In fact, Theorem \ref{theorem: main theorem} says that $\Psi_{N,t}$ can be approximated by $\bigwedge_{j=1}^N \varphi_j^t \otimes W( N^{2/3} \alpha^t) \Omega$, 
where $(\varphi_1^t, \ldots, \varphi_N^t,\alpha^t)$ solves the Schr\"odinger-Klein-Gordon equations
\begin{align}\label{effective_eqs}
\left\{\begin{array}{rl} N^{-1/3} i \partial_t \varphi_j^t(x) \!\!\!\!&= \left( -N^{-2/3}  \Delta + \vphc(x,t)  \right) \varphi_j^t(x),  \qquad \text{for~} j = 1, \ldots, N, \\ 
i \partial_t \alpha^t(k) \!\!\!\!&= N^{-1/3} \delta_N \omega(k) \alpha^t(k) + N^{-1} (2 \pi)^{3/2} \tilde{\eta}(k) \mathcal{F}\left[ \rho^t \right](k), \\
\vphc(x,t) \!\!\!\!&= \int d^3k \, \tilde{\eta}(k) \big(  e^{ikx} \alpha^t(k)  + e^{-ikx} \overline{\alpha^t(k)}  \big),
\end{array}\right.
\end{align}
with $\rho^t = \sum_{i=1}^N \abs{\varphi_i^t}^2$, $(\varphi_1^0, \ldots, \varphi_N^0,\alpha^0) \in (L^2(\mathbb{R}^3))^{N+1}$, $\varphi_1^0,\ldots,\varphi_N^0$ orthonormal, and where $\mathcal{F}[f](k) := (2\pi)^{-3/2}\int d^3x e^{-ikx} f(x)$ denotes the Fourier transform of $f\in L^2(\RRR^3)$. This system of equations is formally equivalent to
\begin{align}\label{eq: Nelson Schroedinger-Klein-Gordon system 2}
i N^{-1/3} \partial_t \varphi_j^t(x) &= \left[ -N^{-2/3} \Delta + \vphc(x,t)  \right] \varphi_j^t(x),  \quad \text{for~} j = 1, \ldots, N, \nonumber\\
\big[ \partial_t^2 + N^{-2/3} \delta_N^2 (- \Delta + m^2)  \big] \vphc(x,t)
&= - N^{-1/3} \delta_N (2 \pi)^{-3/2} \int d^3k \, e^{ikx} \id_{\abs{k}\leq \Lambda}(k) \frac{1}{N}
\mathcal{F}[\rho^t](k).
\end{align}
Its solutions have nice regularity properties because of the ultraviolet cutoff in the radiation field. For $m \in \mathbb{N}$, let $H^m(\mathbb{R}^3)$ denote the Sobolev space of order $m$ and $L_m^2(\mathbb{R}^3)$ a weighted $L^2$-space with norm
$\norm{\alpha}_{L^2_m(\mathbb{R}^3)} =  \big\| ( 1 + \abs{\cdot}^2 )^{m/2} \alpha \big\|_{L^2(\mathbb{R}^3)}$. 
Throughout this paper, we use 
\begin{proposition}
\label{theorem: solution theory}
Let $(\varphi_1^0, \ldots, \varphi_N^0, \alpha^0) \in  \bigoplus_{n=1}^N H^2(\mathbb{R}^3) \oplus L_1^2(\mathbb{R}^3)$. Then there is a strongly differentiable $\bigoplus_{n=1}^N H^2(\mathbb{R}^3) \oplus L_1^2(\mathbb{R}^3)$-valued function $(\varphi_1^t, \ldots, \varphi_N^t, \alpha^t)$ on $[0, \infty )$ that satisfies \eqref{effective_eqs}. Moreover, if $\varphi_1^0,\ldots,\varphi_N^0$ are orthonormal, then so are $\varphi_1^t,\ldots,\varphi_N^t$ for all $t\in[0,\infty)$.
\end{proposition}

\begin{proof}
The proposition can be shown by a standard fixed point argument because of the ultraviolet cutoff. A proof is given in Appendix~\ref{section: Appendix B: The fermionic Schr\"odinger-Klein-Gordon equations}.
\end{proof}

For global well-posedness results of the Schr\"odinger-Klein-Gordon system without UV cutoff, i.e., \eqref{eq: Nelson Schroedinger-Klein-Gordon system 2} with $j=1$, $m=1$ and $\Lambda= \infty$, we refer to \cite{colliander_2008,pecher}.

In order to see that the effective equations are indeed non-trivial and to make the connection to the Coulomb potential, it is instructive to write them explicitly with physical constants. For $m=0$, fermion mass $m_F > 0$, and for $\Lambda = \infty$, \eqref{eq: Nelson Schroedinger-Klein-Gordon system 2} is
\begin{align}
\begin{split}
i(N^{-1/3} \hbar) \partial_t \varphi_j^t(x) &= \left[ -\frac{(N^{-1/3}\hbar)^2}{2m_F} \Delta + \Phi(x,t) \right] \varphi_j^t(x), \\
\left[ \frac{1}{c^2} \partial_t^2 - (N^{-1/3}\delta_N)^2\Delta \right] \Phi(x,t) &= - (N^{-1/3}\delta_N) \frac{e^2}{\varepsilon_0} N^{-1} \rho^t(x).
\end{split}
\end{align}
For $\delta_N=N^{1/3}$ and in the limit $c\to\infty$ this becomes the Poisson equation with solution $\Phi(x,t) = - N^{-1} \frac{e^2}{4\pi \varepsilon_0} (|\cdot|^{-1} * \rho^t)(x)$. Finally, note that in \eqref{effective_eqs} one can write the equation for $\alpha^t(k)$ in integral form and plug it into the equations for the electrons. For $m=0$, $m_F > 0$, and $\Lambda = \infty$, this yields
\begin{align}\label{Hartree_SI_units}
&i (N^{-1/3} \hbar) \partial_t\varphi_j^t(x) \nonumber\\
&\quad = \Bigg[-\frac{(N^{-1/3}\hbar)^2}{2m_F} \Delta + \frac{c\hbar e}{\sqrt{\varepsilon_0}} \, \vphc^{\mathrm{free}}(x,t) \nonumber\\
&\quad\qquad - N^{-2/3}\delta_N^{-1} \frac{e^2}{4\pi \varepsilon_0} \int d^3y\, \frac{1}{|x-y|} \rho^{t-c^{-1}\delta_N^{-1}N^{1/3}|x-y|}(y) \id_{c^{-1}\delta_N^{-1}N^{1/3}|x-y| \leq t} \Bigg] \varphi_j^t(x),
\end{align}
where $\vphc^{\mathrm{free}}(x,t) = e^{-i c \delta_N N^{-1/3} |\nabla| t} \vphc(x,0)$. For $\vphc(x,0)=0$, $\delta_N=N^{1/3}$, and in the formal limit $c\to \infty$, this becomes the Hartree equation with attractive mean-field Coulomb potential.

\section{Main Result}
\label{section: Nelson Main result}

As mentioned above, our goal is to show that $\Psi_{N,t} \approx \bigwedge_{j=1}^N \varphi_j^t \otimes W( N^{2/3} \alpha^t) \Omega$ holds during the time evolution. In the following, this will be proven in the trace-norm distance of reduced density matrices. Let us introduce the number operator
\begin{align}
\label{eq: Nelson number operator definition}
\mathcal{N} \coloneqq \int d^3k \, a^*(k) a(k)
\end{align}
with domain
\begin{align}
\label{eq: Nelson number operator domain}
\mathcal{D}(\mathcal{N}) = \Big\{ \Psi_N \in \mathcal{H}^{(N)} : \sum_{n=1}^{\infty}  n^2 \int d^{3N}x \, d^{3n}k \,  \abs{\Psi_N^{(n)}(X_N, K_n)}^2  < \infty   \Big\}.
\end{align}
Moreover, we choose $\norm{\Psi_{N,0}}=1$ and $\Psi_{N,0} \in \mathcal{H}^{(N)} \cap \mathcal{D}(\mathcal{N}) \cap \mathcal{D}(\mathcal{N}H_N)$ (note that for the definition of the reduced density matrix below we only need $\Psi_{N,0} \in \mathcal{H}^{(N)} \cap \mathcal{D}(\mathcal{N}^{1/2})$). By unitarity also $\norm{\Psi_{N,t}}=1$ and the following lemma holds.

\begin{lemma}\label{solution_th_Schr}
Let $\Psi_{N,0} \in \mathcal{H}^{(N)} \cap \mathcal{D}(\mathcal{N}) \cap \mathcal{D}(\mathcal{N}H_N)$ and let $\Psi_{N,t}$ be the solution to \eqref{eq:Schroedinger_equation_microscopic} with initial condition $\Psi_{N,0}$. Then also $\Psi_{N,t} \in \mathcal{H}^{(N)} \cap \mathcal{D}(\mathcal{N}) \cap \mathcal{D}(\mathcal{N}H_N)$ for all $t\in[0,\infty)$.
\end{lemma}

\begin{proof}
A proof has been given before in \cite{falconi2,falconi} and \cite[Appendix~2.11]{leopold2}.
\end{proof}

For $k\in\NNN$, we define the $k$-particle reduced density matrices of the fermions (as operators on $L^2(\mathbb{R}^{3k})$) by
\begin{align}
\label{eq: Nelson definition reduced one-particle matrix charged particles}
\gamma_{N,t}^{(k,0)} \coloneqq \tr_{k+1,\ldots, N} \tr_{\mathcal{F}_s} \ket{\Psi_{N,t}} \bra{\Psi_{N,t}},
\end{align}
where $\tr_{k+1,\ldots, N}$  denotes the partial trace over the coordinates $x_{k+1},\ldots, x_N$ and $\tr_{\mathcal{F}_s}$ the trace over Fock space. Additionally, we consider on $L^2(\mathbb{R}^3)$ the one-particle reduced density matrix of the bosons with kernel
\begin{align}
\label{eq: Nelson definition reduced one-particle matrix photon}
\gamma_{N,t}^{(0,1)}(k,k') \coloneqq N^{-4/3}  \scp{\Psi_{N,t}}{ a^*(k')  a(k) \Psi_{N,t}}.
\end{align}
The operator $\gamma_{N,t}^{(0,1)}$ is trace class with $\tr\, \gamma_{N,t}^{(0,1)} = N^{-4/3} \scp{\Psi_{N,t}}{\mathcal{N} \Psi_{N,t}}$. It is worth noting that \eqref{eq: Nelson definition reduced one-particle matrix photon} differs from the usual definition $\scp{\Psi_{N,t}}{\mathcal{N} \Psi_{N,t}}^{-1}  \scp{\Psi_{N,t}}{ a^*(k')  a(k) \Psi_{N,t}}$, which has trace one. In our choice we only measure deviations from the classical mode function that are at least of order $N^{4/3}$. This is important if one starts initially with no bosons and examines the one-particle reduced density matrix after short times when only a few bosons have been created. Then,  the state of the bosons might not be coherent and the usual definition of the one-particle reduced density matrix may not converge to the classical mode function. However, such mismatches are not important for the dynamics (and hence neglected in our definition) because the field operator is rescaled by a factor of $N^{-2/3}$, see \eqref{eq:Schroedinger_equation_microscopic2}.

Let us now state the main result of this article. We summarize the conditions on our initial data in the following assumption. We denote the trace norm of an operator $A$ by $\norm{A}_{\tr} := \tr \abs{A}$.

\begin{assumption}\label{main_assumption}
We have $\alpha^0 \in L_1^2(\mathbb{R}^3)$ and $\varphi_1^0, \ldots, \varphi_N^0 \in H^2(\mathbb{R}^3)$ orthonormal and such that
\begin{align}\label{eq: theorem semiclassical structure}
\norm{p^0 e^{ikx} q^0}_{\tr} \leq C (1 + \abs{k}) N^{2/3}~\forall k\in\RRR^3
\quad \text{and} \quad
\norm{p^0 \nabla q^0}_{\tr} \leq C N
\end{align}
for some $C>0$, where $p^t = \sum_{j=1}^N \ketbr{\varphi_j^t}$ and $q^t=1-p^t$ for any $t\in\RRR$ (see also Definition~\ref{def_of_projs}). Moreover, $\Psi_{N,0} \in \mathcal{H}^{(N)} \cap \mathcal{D}(\mathcal{N}) \cap \mathcal{D}(\mathcal{N}H_N)$ with $\norm{\Psi_{N,0}} =1$.
\end{assumption}

Our main theorem is the following.
\begin{theorem}\label{theorem: main theorem}
Let Assumption~\ref{main_assumption} hold and let $\Psi_{N,t}$ be the solution to \eqref{eq:Schroedinger_equation_microscopic} with initial condition $\Psi_{N,0}$, and $\varphi_1^t,\ldots,\varphi_N^t,\alpha^t$ the solution to \eqref{effective_eqs} with initial condition $\varphi_1^0,\ldots,\varphi_N^0,\alpha^0$. We define
\begin{align}
a_N &= \norm{\gamma_{N,0}^{(1,0)} - N^{-1}p^0}_{\tr},
\\
b_N &= N^{1/3} \, \tr \big( \gamma_{N,0}^{(2,0)}   q^0 \otimes q^0 \big),
\\
c_N &= N^{-1} \scp{W^{-1}(N^{2/3} \alpha^0) \Psi_{N,0}}{ \mathcal{N} W^{-1}( N^{2/3} \alpha^0) \Psi_{N,0}}.
\end{align}
Then there exists $C>0$ (independent of $N$, $\delta_N$, $\Lambda$, and $t$) such that for any $t \geq 0$,
\begin{align}\label{eq: Nelson main theorem 1}
\norm{\gamma_{N,t}^{(1,0)} - N^{-1} p^t}_{\tr} &\leq
\sqrt{a_N + b_N + c_N + N^{-1}} \, e^{e^{C \Lambda^4 (1 + \|\alpha^0\|_2)(1 + t^2)}}.
\end{align}
If additionally $c_N \leq \tilde{C}N^{1/3}$ for some $\tilde{C}>0$, then
\begin{align}\label{eq: Nelson main theorem 2}
\norm{\gamma_{N,t}^{(0,1)} - \ket{\alpha^t} \bra{\alpha^t}}_{\tr} &\leq
\sqrt{ N^{-1/3} (a_N + b_N + c_N) + N^{-4/3}} \, e^{e^{C \Lambda^4 ( 1 + \|\alpha^0\|_2) ( 1 + t^2)}}.
\end{align}
In particular, for $\Psi_{N,0} = \bigwedge_{j=1}^N \varphi_j^0 \otimes W( N^{2/3} \alpha^0) \Omega$ we have $a_N=b_N=c_N=0$ and one obtains
\begin{align}\label{eq: Nelson main theorem 3}
\norm{\gamma_{N,t}^{(1,0)} -  N^{-1} p^t}_{\tr} &\leq
N^{-1/2}  e^{e^{C \Lambda^4 (1 + \|\alpha^0\|_2)(1 + t^2)}} , \\
\label{eq: Nelson main theorem 4}
\norm{\gamma_{N,t}^{(0,1)} - \ket{\alpha^t} \bra{\alpha^t}}_{\tr} &\leq
N^{-2/3} e^{e^{C \Lambda^4 (1 + \|\alpha^0\|_2)(1 + t^2)}}.
\end{align}
\end{theorem}
The theorem is proved in Section~\ref{sec_main_proof}.

\begin{remark}
In \cite{erdoes:2004} and \cite{benedikter_2013} a similar limit was considered for fermions that interact by means of a pair potential. From these works we learned the importance of the semiclassical structure. The most related works from a technical point of view are \cite{petrat_2015} and \cite{falconi, leopold, leopold3}.
\end{remark}

\begin{remark}\label{remark_sc_struc}
For initial states without semiclassical structure, i.e., without assuming \eqref{eq: theorem semiclassical structure}, the result only holds true for times of order $N^{-1/3}$. More precisely, Equations~\eqref{eq: Nelson main theorem 1}--\eqref{eq: Nelson main theorem 4} hold with the double exponential replaced by $e^{C(\Lambda,\norm{\alpha^0})N^{1/3}t}$. 

The first inequality in \eqref{eq: theorem semiclassical structure} means that the kernel $p^0(x,y)$ is localized around a distance smaller than of order $N^{-1/3}$ around the diagonal $x=y$. The second inequality means that the density varies on scales of order one. In fact, these conditions should imply that the time evolution of $p^0$ (or, say, its Wigner transform) is close to a classical evolution equation, which here is the Vlasov equation. This has indeed been shown in the two-body interaction case, let us refer to \cite{benedikter:2015} and references therein. Note also that for simple cases like plane waves in a box of volume of order one, \eqref{eq: theorem semiclassical structure} indeed holds, see \cite{porta:2016}. For a more thorough discussion of these conditions we refer to \cite{benedikter_2013,porta:2016}.
\end{remark}

\begin{remark}
Let us give a bit more intuition about $c_N$. We first note that the Weyl operator, defined in \eqref{eq: Nelson Weyl operator}, is unitary, and thus $W^{-1}(f) = W^*(f) = W(-f)$. One of its well-known properties (see, e.g., \cite{rodnianskischlein} for a nice exposition) is 
\begin{equation}
W^*(f) a(k) W(f) = a(k) + f(k) , \quad  W^*(f) a^*(k) W(f) = a^*(k) + \overline{f(k)}.
\end{equation}
With that at hand we can write $c_N$ as
\begin{align}\label{beta_b_c}
c_N &= N^{1/3} \int d^3k\, \norm{N^{-2/3} a(k) W^{-1}(N^{2/3} \alpha^0) \Psi_{N,0}}^2 \nonumber\\
&= N^{1/3} \int d^3k\, \norm{\left( N^{-2/3} a(k) - \alpha^0(k) \right) \Psi_{N,0}}^2,
\end{align}
from which it might become more clear that $c_N$ measures the initial deviations around the classical radiation field $\alpha^0$.
\end{remark}

In the case of $\delta_N = N^{1/3 - \epsilon}$ with $\epsilon > 0$, the group velocity of the bosons is of lower order than the average speed of the electrons.
This implies that the electrons effectively experience a stationary scalar field and  evolve according to 
\begin{align}
\label{eq: free evolution}
N^{-1/3} i \partial_t \varphi_j^t(x)
&= \left(  -  N^{-2/3}  \Delta + \vphc(x,0) \right) \varphi_j^t(x) 
\quad \text{for} \; j=1, \ldots, N .
\end{align}
The precise statement is the following.
\begin{theorem}
\label{theorem: free evolution}
Let Assumption \ref{main_assumption} hold, let $(\varphi_1^t, \ldots, \varphi_N^t, \alpha^t)$ be the solution to \eqref{effective_eqs} with initial condition $(\varphi_1^0, \ldots, \varphi_N^0, \alpha^0)$ and let $(\el_1^t, \ldots, \el_N^t)$ be the solution to \eqref{eq: free evolution}  with initial condition $(\varphi_1^0, \ldots, \varphi_N^0)$.
We define $\tilde{p}^t = \sum_{j=1}^N \ketbr{\el_j^t}$ and $p^t$ as in Assumption \ref{main_assumption}.
Then there exists $C >0$ (independent of $N$, $\delta_N$, $\Lambda$, and $t$) such that 
\begin{align}\label{eq:free_dyn_to_Hartree}
N^{-1} \norm{ p^t - \tilde{p}^t}_{\tr} 
&\leq  N^{-1/3} \delta_N e^{C \Lambda^4 (1 + \norm{\alpha^0}_2) ( 1 + t^2)} .
\end{align}
Furthermore, let $\Psi_{N,t}$ be the solution to \eqref{eq:Schroedinger_equation_microscopic} with initial condition $\Psi_{N,0}$, and let $a_N$, $b_N $ and $c_N$ be defined as in Theorem~\ref{theorem: main theorem}. Then there exists $C>0$ (independent of $N$, $\delta_N$, $\Lambda$, and $t$) such that for all $t \geq 0$,
\begin{align}\label{eq:free_dyn_to_Schr}
\norm{\gamma_{N,t}^{(1,0)} - N^{-1} \tilde{p}^t}_{\tr} &\leq
\Big(  N^{-1/3} \delta_N   +
\sqrt{a_N + b_N + c_N + N^{-1}} \Big)  \, e^{e^{C \Lambda^4 (1 + \|\alpha^0\|_2)(1 + t^2)}}.
\end{align}
\end{theorem}
The theorem is proved in Appendix~\ref{section:convergence_free evolution}.

\section{Structure of the Proof}

In order to prove Theorem \ref{theorem: main theorem}, it is important to define and control the right macroscopic variables. For that, we adapt techniques that are based on the method from \cite{pickl1} and that were further developed in \cite{petrat_2015,leopold,leopold3}. In addition, it is crucial to find the right measure for the correlations between the electrons and to consider only initial states with semiclassical structure.
The key idea of the proof is to define a suitable functional $\beta(\Psi_N,\varphi_1, \ldots, \varphi_N, \alpha)$ which measures if the fermions are close to an antisymmetrized product state $\bigwedge_{j=1}^N \varphi_j$ with $\varphi_1, \ldots, \varphi_N$ orthonormal and if the state of the radiation field is approximately coherent.
To this end, we introduce the following operators.
\begin{definition}\label{def_of_projs}
For $N \in \mathbb{N}$, $m, j \in \{1,2, \ldots, N \}$ and $\varphi_1, \ldots, \varphi_N \in L^2(\mathbb{R}^3)$ orthonormal we define the projectors
$p_m^{\varphi_j}: L^2(\mathbb{R}^{3N}) \rightarrow  L^2(\mathbb{R}^{3N})$ by
\begin{align}
p_m^{\varphi_j} f(x_1, \ldots, x_N)
&\coloneqq  \varphi_j(x_m) \int d^3x_m \, \overline{\varphi_j(x_m)} f(x_1, \ldots, x_N) 
\quad  \forall \, f  \in L^2(\mathbb{R}^{3N}).
\end{align}
Moreover, we define the projectors
$p_m^{\varphi_1, \ldots, \varphi_N}: L^2(\mathbb{R}^{3N}) \rightarrow  L^2(\mathbb{R}^{3N})$ and
$q_m^{\varphi_1, \ldots, \varphi_N}: L^2(\mathbb{R}^{3N}) \rightarrow  L^2(\mathbb{R}^{3N})$  by
\begin{align} 
p_m^{\varphi_1, \ldots, \varphi_N} &\coloneqq \sum_{j=1}^N p_m^{\varphi_j}
\quad \text{and} \quad
q_m^{\varphi_1, \ldots, \varphi_N} \coloneqq 1 - p_m^{\varphi_1, \ldots, \varphi_N}.
\end{align}
\end{definition}

\noindent
The correlations between the electrons are controlled by means of two functionals.

\begin{definition}
Let $N \in \mathbb{N}$, $\varphi_1, \ldots, \varphi_N \in L^2(\mathbb{R}^3)$ orthonormal and $\Psi_{N} \in  \mathcal{H}^{(N)}$. Then, $\beta^{a,1}: \mathcal{H}^{(N)} \times  L^2(\mathbb{R}^3)^{N} \rightarrow [0, \infty)$ and 
$\beta^{a,2}: \mathcal{H}^{(N)} \times  L^2(\mathbb{R}^3)^N  \rightarrow [0, \infty)$ are given by 
\begin{align}
\beta^{a,1}(\Psi_{N}, \varphi_1, \ldots, \varphi_N) &\coloneqq
\scp{\Psi_{N}}{q_1^{\varphi_1, \ldots, \varphi_N} \otimes \id_{\mathcal{F}_s}  \Psi_{N,t}} 
\quad \text{and}
\\
\beta^{a,2}(\Psi_{N}, \varphi_1, \ldots, \varphi_N) &\coloneqq
 N^{1/3} \scp{\Psi_{N}}{q_1^{\varphi_1, \ldots, \varphi_N} q_2^{\varphi_1, \ldots, \varphi_N} \otimes \id_{\mathcal{F}_s}  \Psi_{N}}.
\end{align}
\end{definition}
We note that $\beta^{a,1}(\Psi_{N}, \varphi_1, \ldots, \varphi_N)$ corresponds to the expectation value of the relative number of fermions outside the antisymmetric product $\bigwedge_{j=1}^N\varphi_j$ (i.e., the number of excitations around the state $\bigwedge_{j=1}^N\varphi_j$ divided by $N$). The functional $N^{-1/3}\beta^{a,2}(\Psi_{N}, \varphi_1, \ldots, \varphi_N)$ corresponds (up to a small error) to the expectation value of the square of this number. More details about the technical relevance of $\beta^{a,2}$ are given at the beginning of Section~\ref{sec: Estimates on the time derivative}. 

In order to determine whether the state of the radiation field is coherent, we define $\beta^b$, which measures the fluctuations of the field modes around the complex function $\alpha$.
\begin{definition}
Let $\alpha \in L^2(\mathbb{R}^3)$ and $\Psi_{N} \in \mathcal{H}^{(N)} \cap \mathcal{D}\left(\mathcal{N} \right)$. Then $\beta^b: \mathcal{H}^{(N)} \cap \mathcal{D}\left(\mathcal{N} \right) \times L^2(\mathbb{R}^3) \rightarrow [0,\infty)$ is given by
\begin{align}
\beta^b\left( \Psi_{N}, \alpha \right) 
\coloneqq N^{1/3} \int d^3k \, \SCP{\left(N^{-2/3} a(k) - \alpha(k) \right) \Psi_{N}}{ \left( N^{-2/3} a(k) - \alpha(k) \right) \Psi_{N}}.
\end{align} 
\end{definition}
Note that $\beta^b(\Psi_{N,0},\alpha^0) = c_N$ as we showed in \eqref{beta_b_c}. Let us also remark that when $\Psi_{N,t}$ is a solution to \eqref{eq:Schroedinger_equation_microscopic} and $\varphi_1^t,\ldots,\varphi_N^t,\alpha^t$ a solution to \eqref{effective_eqs}, then the functional $\beta^b\left( \Psi_{N,t}, \alpha^t \right)$ coincides (up to scaling) with the one used in the coherent states approach, see, e.g., \cite[Chapter~3]{schleinbook}. Finally, the functional $\beta$ is defined by
\begin{definition}
Let $N \in \mathbb{N}$, $\varphi_1, \ldots, \varphi_N \in L^2(\mathbb{R}^3)$ orthonormal, $\alpha \in L^2(\mathbb{R}^3)$ and $\Psi_{N} \in \mathcal{H}^{(N)} \cap \mathcal{D}\left(\mathcal{N} \right)$. Then $\beta: \mathcal{H}^{(N)} \cap \mathcal{D}\left(\mathcal{N} \right) \times L^2(\mathbb{R}^3)^N \times L^2(\mathbb{R}^3) \rightarrow [0, \infty)$ is 
defined by
\begin{align}
\beta \left( \Psi_{N}, \varphi_1, \ldots, \varphi_N, \alpha \right) 
&\coloneqq \beta^{a,1}(\Psi_{N}, \varphi_1, \ldots, \varphi_N)
+ \beta^{a,2}(\Psi_{N}, \varphi_1, \ldots, \varphi_N)
+  \beta^b\left( \Psi_{N}, \alpha \right).
\end{align}
\end{definition}

\noindent
In the following, we are interested in the value of $\beta \left( \Psi_{N,t}, \varphi_1^t, \ldots, \varphi_N^t, \alpha^t \right) $, where $(\varphi_1^t, \ldots, \varphi_N^t, \alpha^t)$ is a solution to the Schr\"odinger-Klein-Gordon equations \eqref{effective_eqs} and $\Psi_{N,t}$ evolves according to the Schr\"odinger equation \eqref{eq:Schroedinger_equation_microscopic}.  In this case, we apply the shorthand notations $\beta(t)$, $\beta^{a,1}(t)$, $\beta^{a,2}(t)$ and $\beta^b(t)$. Moreover, we use the abbreviations $p_m^{\varphi_1^t, \ldots, \varphi_N^t} = p_m^t$ , $q_m^{\varphi_1^t, \ldots, \varphi_N^t} = q_m^t $ and write $p_m^{\varphi_j^t}$ occasionally as $\ketbr{\varphi_j^t}_m$.

For the proof of Theorem~\ref{theorem: main theorem} we pursue the following strategy.

\begin{itemize}

\item[A)] We choose initial data $(\varphi_1^0, \ldots, \varphi_N^0, \alpha^0)$ of the Schr\"odinger-Klein-Gordon system \eqref{effective_eqs} and a many-body wave function $\Psi_{N,0}$ that satisfy our Assumption~\ref{main_assumption}. Proposition~\ref{theorem: solution theory} and Lemma~\ref{solution_th_Schr} make sure that the solutions at any time $t\geq 0$ are regular enough, and in Section~\ref{section_semiclassical} we show that the solutions still have the semiclassical structure.

\item[B)] After that, we control the change of $\beta(t)$ in time. For this, we use the semiclassical structure to estimate $\abs{\frac{d}{dt} \beta(t)} \leq e^{Ct} \left( \beta(t) + N^{-1} \right)$ for some $C>0$ at each time $t \geq 0$. Gronwall's lemma then yields $\beta(t) \leq e^{e^{Ct}} \left( \beta(0) + N^{-1} \right)$.

\item[C)] Finally, we relate the initial states of Theorem \ref{theorem: main theorem} and the trace norm convergence of the reduced density matrices to $\beta(t)$.

\end{itemize}

\begin{notation}
In the rest of this article the letter $C$ denotes a generic positive constant and its value might change from line to line for notational convenience.
\end{notation}

\section{Semiclassical Structure}\label{section_semiclassical}

We first prove that the semiclassical structure from Equation~\eqref{eq: theorem semiclassical structure} can be propagated in time. The Hilbert-Schmidt norm of an operator $A$ is denoted by $\norm{A}_{\HS} := \sqrt{\tr \, A^*A}$.

\begin{lemma}
\label{lemma: semiclassical structure}
Let $(\varphi_1^0, \ldots, \varphi_N^0, \alpha^0) \in H^2(\mathbb{R}^3)^N \times L_1^2(\mathbb{R}^3)$ with orthonormal $\varphi_1^0,\ldots,\varphi_N^0$ and let $(\varphi_1^t, \ldots, \varphi_N^t, \alpha^t)$ be solutions of the equations \eqref{effective_eqs}. We assume that
\begin{align}
\norm{p^0e^{ikx}q^0}_{\tr} \leq \tilde{C} (1+\abs{k}) N^{2/3}
\end{align}
for all $k\in\RRR^3$ and
\begin{align}
\norm{p^0\nabla q^0}_{\tr} \leq \tilde{C} N
\end{align}
for some $\tilde{C}>0$. Then there exists some $C>0$ (independent of $N$, $\Lambda$ and $t$) such that
\begin{align}
\label{eq: semiclassical structure}
\norm{p^te^{ikx}q^t}_{\HS}^2 \leq \norm{p^te^{ikx}q^t}_{\tr} \leq 2\tilde{C} (1+\abs{k}) N^{2/3} e^{C \Lambda^4  ( 1 + \norm{\alpha^0}_{2}) (1 + t^2)}
\end{align}
for all $k\in\RRR^3$ and
\begin{align}
\norm{p^t\nabla q^t}_{\tr} \leq 2\tilde{C} N e^{C \Lambda^4 ( 1 + \norm{\alpha^0}_{2}) (1 + t^2)}
\end{align}
for all $t\in \RRR$.
\end{lemma}

\begin{remark}
\label{rem: comm relation semiclassical structure}
We could formulate Lemma \ref{lemma: semiclassical structure} likewise in terms of $\norm{ \left[ p^t, e^{ikx} \right]}_{\tr} $ and $\norm{\left[ p^t, \nabla  \right]}_{\tr}$ as was done in \cite{benedikter_2013}, because
\begin{align}
\begin{split}
\norm{p^t e^{ikx} q^t}_{\tr} = \norm{\left[p^t, e^{ikx}\right] q^t}_{\tr} &\leq \norm{ \left[ p^t, e^{ikx} \right]}_{\tr} 
\leq \norm{p^t e^{ikx} q^t}_{\tr} + \norm{p^t e^{-ikx} q^t}_{\tr}, \\
\norm{p^t \nabla q^t}_{\tr} = \norm{\left[p^t, \nabla\right] q^t}_{\tr} &\leq \norm{\left[ p^t, \nabla \right]}_{\tr} 
\leq 2 \norm{p^t \nabla q^t}_{\tr}.
\end{split}
\end{align}
These inequalities hold since $p^t q^t = 0$, $\norm{AB}_{\tr} \leq \norm{A}\norm{B}_{\tr}$ and $\norm{BA}_{\tr} \leq \norm{A}\norm{B}_{\tr}$ for $A$ bounded and $B$ trace class, $\norm{q^t}=1$, and $\norm{B}_{\tr} = \norm{B^*}_{\tr}$ for $B$ trace class.
\end{remark}

\begin{proof}[Proof of Lemma \ref{lemma: semiclassical structure}]
The propagation of the semiclassical structure is shown in a similar way as in \cite[Section~5]{benedikter_2013}. Recall that due to Proposition~\ref{theorem: solution theory} the solution $(\varphi_1^t, \ldots, \varphi_N^t, \alpha^t)$ is in $\bigoplus_{n=1}^N H^2(\mathbb{R}^3) \oplus L_1^2(\mathbb{R}^3)$ and strongly continuous. If we define $h^t = - \Delta + N^{2/3} \vphc(\cdot,t)$, the time derivative of the projector $iN^{1/3}\partial_t p^t = [h^t,p^t]$. Then\footnote{Note that with an operator like $p^t \nabla$ we mean the trace class operator $\sum_{j=1}^N \ketbra{\varphi_j^t}{-\nabla\varphi_j^t}$, which is well-defined due to Proposition~\ref{theorem: solution theory}.}
\begin{align}
i N^{1/3} \partial_t \big( q^te^{ikx}p^t \big) = [h^t,q^t]e^{ikx}p^t + q^te^{ikx}[h^t,p^t] = [h^t,q^te^{ikx}p^t] - q^t[h^t,e^{ikx}]p^t.
\end{align}
From
\begin{align}
[h^t,e^{ikx}] = [- \Delta, e^{ikx}] &= - ik \Big(\nabla e^{ikx} + e^{ikx}\nabla \Big)
\end{align}
and using $p^t+ q^t = 1$, we conclude
\begin{align}\label{sc_time_der}
i N^{1/3} \partial_t \big(q^te^{ikx}p^t\big) &= [h^t,q^te^{ikx}p^t] +ik q^t\Big(\nabla e^{ikx} + e^{ikx}\nabla \Big)p^t \nonumber\\
&= [h^t,q^te^{ikx}p^t] + ik \nabla q^te^{ikx}p^t + ik q^te^{ikx}p^t \nabla \nonumber\\
&\quad + ik \Big( q^t\nabla p^t e^{ikx}p^t - p^t\nabla q^te^{ikx}p^t - q^te^{ikx} p^t\nabla q^t + q^te^{ikx}q^t\nabla p^t \Big) \nonumber\\
&= \Big(h^t + ik\nabla\Big) q^te^{ikx}p^t - q^te^{ikx}p^t \Big(h^t - ik\nabla\Big) \nonumber\\
&\quad + ik \Big( \big(q^t\nabla p^t - p^t\nabla q^t\big) e^{ikx}p^t + q^te^{ikx} \big(q^t\nabla p^t - p^t\nabla q^t\big) \Big).
\end{align}
Next, we define the time dependent self-adjoint operators
\begin{align}
A_{+k}(t) = h^t + ik \nabla 
\quad \text{and}   \quad
A_{-k}(t) = h^t - ik \nabla
\end{align}
and their respective unitary propagators $U_{+k}(t;s)$ and $U_{-k}(t;s)$. These are indeed well-defined, which follows from \cite[Theorem~X.71]{reedsimon} adapted to $H_0 = -\Delta \pm i\nabla k$, or, more conveniently, from \cite[Theorem~2.5]{mcgriesemer} and the fact that $\vphc(\cdot,t)$ is continuously differentiable in $L^{\infty}(\RRR^3)$, a direct consequence of Proposition~\ref{theorem: solution theory}. The unitary propagators (with rescaled time) satisfy
\begin{align}
i N^{1/3} \partial_t U_{+k}(t;s)\varphi = A_{+k}(t) U_{+k}(t;s) \varphi
\quad \text{and} \quad
i N^{1/3} \partial_t U_{-k}(t;s)\varphi = A_{-k}(t) U_{-k}(t;s) \varphi
\end{align}
for all $\varphi \in H^2(\RRR^3)$, with initial conditions $U_{+k}(s;s)= U_{-k}(s;s) = 1$.
This gives
\begin{align}
& i N^{1/3} \partial_t \big(U_{+k}^*(t;0) q^te^{ikx}p^t U_{-k}(t;0)\big) \nonumber\\
&\qquad = U_{+k}^*(t;0) \Big( - A_{+k}(t)  q^te^{ikx}p^t + q^te^{ikx}p^t A_{-k} + i N^{1/3} \partial_t \big( q^te^{ikx}p^t \big)\Big) U_{-k}(t;0) \nonumber\\
&\qquad = ik U_{+k}^*(t;0) \Big( \big(q^t\nabla p^t - p^t\nabla q^t\big)e^{ikx}p^t + q^te^{ikx}\big( q^t\nabla p^t \- p^t\nabla q^t \big) \Big)  U_{-k}(t;0),
\end{align}
which leads to 
\begin{align}
&U_{+k}^*(t;0) q^te^{ikx}p^t U_{-k}(t;0) = q^0e^{ikx}p^0 \nonumber\\
&\qquad + N^{-1/3} k \int_0^t ds \, U_{+k}^*(s;0) 
\bigg( \big(q^s\nabla p^s - p^s\nabla q^s\big) e^{ikx}p^s - q^se^{ikx}\big(p^s\nabla q^s - q^s\nabla p^s\big)\bigg) U_{-k}(s;0)
\end{align}
and thus
\begin{align}
&q^te^{ikx}p^t = U_{+k}(t;0) q^0e^{ikx}p^0 U_{-k}^*(t;0) \nonumber\\
&\qquad+ N^{-1/3} k \int_0^t ds \, U_{+k}(t;s) \bigg( \big(q^s\nabla p^s - p^s\nabla q^s\big) e^{ikx}p^s - q^se^{ikx}\big(p^s\nabla q^s - q^s\nabla p^s\big)\bigg) U_{-k}(s;t) .
\end{align}
For the trace norm, we then obtain the estimate
\begin{align}\label{a_tr_est}
\norm{q^te^{ikx}p^t}_{\tr}
&\leq \norm{q^0e^{ikx}p^0}_{\tr}
+ 4 N^{-1/3} (1+ \abs{k}) \int_0^t ds \, \norm{q^s\nabla p^s}_{\tr},
\end{align}
where we used that $\norm{AB}_{\tr} \leq \norm{A}\norm{B}_{\tr}$ and $\norm{BA}_{\tr} \leq \norm{A}\norm{B}_{\tr}$ for $A$ bounded and $B$ trace class.
Thus,
\begin{align}
\label{eq: semicl. comm. 1}
\sup_{k \in \mathbb{R}^3} \left((1 + \abs{k})^{-1} \norm{q^te^{ikx}p^t}_{\tr} \right)
&\leq \sup_{k \in \mathbb{R}^3} \left((1 + \abs{k})^{-1} \norm{q^0e^{ikx}p^0}_{\tr} \right) \nonumber\\
&\quad + 4 \int_0^t ds \, N^{-1/3} \norm{q^s\nabla p^s}_{\tr}.
\end{align}
In order to control the latter term, we calculate the time derivative of $q^t\nabla p^t$. We find 
\begin{align}
i N^{1/3} \partial_t \big(q^t\nabla p^t\big)
&= [h^t,q^t]\nabla p^t + q^t \nabla [h^t,p^t] \nonumber\\
&= [h^t,q^t\nabla p^t] - q^t [h^t,\nabla]p^t \nonumber\\
&= [h^t,q^t\nabla p^t] + N^{2/3} q^t (\nabla \vphc)(t) p^t.
\end{align}
In analogy to the previous calculation, we define the two-parameter group $U_h(t;s)$ satisfying
\begin{align}
i N^{1/3} \partial_t U_h(t;s)\varphi = h^t U_h(t;s) \varphi
\end{align}
for all $\varphi\in H^2(\RRR^3)$ and $U_h(s;s)=1$. Then, we calculate
\begin{align}
&i N^{1/3} \partial_t \big( U_h^*(t;0)q^t\nabla p^t U_h(t;0)\big) \nonumber\\
&\quad= U_h^*(t;0) \bigg( -h^t q^t\nabla p^t + q^t\nabla p^t h^t + i N^{1/3}\partial_t \big(q^t \nabla p_t\big) \bigg) U_h(t;0) \nonumber\\
&\quad= N^{2/3} U_h^*(t;0) q^t (\nabla \vphc)(t) p^t  U_h(t;0),
\end{align}
which implies
\begin{align}
q^t \nabla p^t &= U_h(t;0) q^0\nabla p^0 U^*_h(t;0) - i N^{1/3} \int_0^t ds \, U_h(t;s) \Big(q^s (\nabla \vphc)(s)p^s\Big) U_h(s,t).
\end{align}
Using the same inequalities as for \eqref{a_tr_est}, this leads to
\begin{align}
\norm{q^t\nabla p^t}_{\tr} &\leq \norm{q^0\nabla p^0}_{\tr} + N^{1/3}  \int_0^t ds \, \norm{q^s(\nabla\vphc)(s)p^s}_{\tr}.
\end{align}
By Lemma~\ref{lemma: growth of alpha-t}, which says that $\norm{\alpha^t}_2 \leq \norm{\alpha^0}_2 +   \norm{\tilde{\eta}}_2 \abs{t}$, we can estimate
\begin{align}
\norm{q^s(\nabla \vphc)(s)p^s}_{\tr} 
&= \norm{\int d^3k \, \tilde{\eta}(k) k \big( \alpha^s(k) q^se^{ikx}p^s - \overline{\alpha^s(k)} q^se^{-ikx}p^s \big)}_{\tr} 
\nonumber\\
&\leq \int d^3k \, \tilde{\eta}(k) \abs{k} \left( \abs{\alpha^s(k)} 
\norm{q^se^{ikx}p^s}_{\tr} + \abs{\alpha^s(k)} \norm{q^se^{-ikx}p^s}_{\tr} \right)
\nonumber\\
&\leq  2 \norm{(1 + \abs{\cdot})^2 \tilde{\eta}}_2 \norm{\alpha^s}_2   \sup_{k \in \mathbb{R}^3} \Big((1 + \abs{k})^{-1} \norm{q^se^{ikx}p^s}_{\tr} \Big)
\nonumber\\
&\leq 2 \norm{(1 +\abs{\cdot})^2 \tilde{\eta}}_2 \Big( \norm{\alpha^0}_2 + \norm{\tilde{\eta}}_2 \abs{s} \Big) \sup_{k \in \mathbb{R}^3} \Big((1 + \abs{k})^{-1} \norm{q^se^{ikx}p^s}_{\tr} \Big)
\end{align}
and obtain 
\begin{align}
N^{-1/3} &\norm{q^t\nabla p^t}_{\tr} 
\leq N^{-1/3} \norm{q^0\nabla p^0}_{\tr} \nonumber\\
&+ 2 \norm{(1 + \abs{\cdot})^2  \tilde{\eta}}_2 \int_0^t ds \,  \Big( \norm{\alpha^0}_2 + \norm{\tilde{\eta}}_2  \abs{s} \Big)
\sup_{k \in \mathbb{R}^3} \Big((1 + \abs{k})^{-1} \norm{q^se^{ikx}p^s}_{\tr} \Big).
\end{align}
Together with the estimate~\eqref{eq: semicl. comm. 1}, this gives
\begin{align}
&\sup_{k \in \mathbb{R}^3} \Big((1 + \abs{k})^{-1} \norm{q^te^{ikx}p^t}_{\tr} \Big) + N^{-1/3} \norm{q^t\nabla p^t}_{\tr} \nonumber \\
&\quad\leq \sup_{k \in \mathbb{R}^3} \Big((1 + \abs{k})^{-1} \norm{q^0e^{ikx}p^0}_{\tr} \Big) + N^{-1/3} \norm{q^0\nabla p^0}_{\tr}
\nonumber\\
&\qquad+ \int_0^t ds \, C(\Lambda,s,\norm{\alpha^0}_2) \left( \sup_{k \in \mathbb{R}^3} \Big((1 + \abs{k})^{-1} \norm{q^se^{ikx}p^s}_{\tr} \Big) + N^{-1/3} \norm{q^s\nabla p^s}_{\tr} \right),
\end{align}
where $C(\Lambda,s,\norm{\alpha^0}_2) \coloneqq 4 + 2 \norm{( 1 + \abs{\cdot})^2 \tilde{\eta}}_2 \big( \norm{\alpha^0}_2 + \norm{\tilde{\eta}}_2 \abs{s} \big)$. By means of Gronwall's lemma and the chosen initial conditions, we obtain
\begin{align}
&\sup_{k \in \mathbb{R}^3} \Big((1 + \abs{k})^{-1} \norm{q^te^{ikx}p^t}_{\tr} \Big) + N^{-1/3} \norm{q^t\nabla p^t}_{\tr}  
\nonumber\\
&\quad\leq 2\tilde{C} N^{2/3} \exp \Big[4 \abs{t} \big( 1 + \norm{( 1 + \abs{\cdot}^2) \tilde{\eta}}_2 \big( \norm{\alpha^0}_2 + \norm{\tilde{\eta}}_2  \abs{t} \big) \Big]
\nonumber \\
&\quad\leq 2\tilde{C} N^{2/3} \exp \Big[ C \Lambda^4 \big( 1 + \norm{\alpha^0}_2 \big) \big( 1 + t^2 \big)  \Big].
\end{align}
Finally, note that
\begin{equation}
\norm{p^te^{ikx}q^t}_{\HS}^2 = \norm{q^te^{-ikx}p^te^{ikx}q^t}_{\tr} \leq \norm{p^te^{ikx}q^t}_{\tr}.
\end{equation}
\end{proof}

\section{Estimates on the Time Derivative}
\label{sec: Estimates on the time derivative}

In this section we control the change of $\beta(t)$ in time by separately estimating the time derivatives of $\beta^{a,1}(t)$, $\beta^{a,2}(t)$ and $\beta^b(t)$. Note that the time derivative of $\beta^{a,1}(t)$ can be controlled in terms of $\beta^{a,1}(t)$ itself, $\beta^{b}(t)$, and an error of order $N^{-1}$. The time derivative of $\beta^{b}(t)$, however, is controlled in terms of $\beta^{a,1}(t)$, $\beta^{a,2}(t)$, $\beta^{b}(t)$ itself, and an error of order $N^{-1}$. This is why we also introduced $\beta^{a,2}(t)$. It allows us to close the Gronwall argument, since its time derivative can be bounded in terms of $\beta^{a,1}(t)$, $\beta^{a,2}(t)$ itself, $\beta^{b}(t)$, and an error of order $N^{-5/3}$. We first compute the corresponding time derivatives. Then, in the following subsections, we bound these expressions as explained above.
\begin{lemma}\label{lemma: time derivative of beta preparation}
Let $\alpha^0 \in L_1^2(\mathbb{R}^3)$, $\varphi_1^0, \ldots, \varphi_N^0 \in H^2(\mathbb{R}^3)$ orthonormal and $\Psi_{N,0} \in \mathcal{H}^{(N)} \cap \mathcal{D}(\mathcal{N}) \cap \mathcal{D}(\mathcal{N}H_N)$ with $\norm{\Psi_{N,0}} =1$. Let $\Psi_{N,t}$ be the solution to \eqref{eq:Schroedinger_equation_microscopic} with initial condition $\Psi_{N,0}$, and $\varphi_1^t,\ldots,\varphi_N^t,\alpha^t$ the solution to \eqref{effective_eqs} with initial condition $\varphi_1^0,\ldots,\varphi_N^0,\alpha^0$. Then
\begin{align}
\label{eq: time derivative of beta-a-1 preparation}
\frac{d}{dt} \beta^{a,1}(t) &= -2 N^{1/3}  \Im \, \scp{\Psi_{N,t}}{p_1^t\Big( N^{-2/3} \vph(x_1) - \vphc(x_1,t) \Big) q^t_1 \Psi_{N,t}},
\\
\label{eq: time derivative of beta-a-2 preparation}
\frac{d}{dt} \beta^{a,2}(t) &= -4 \Im N^{2/3}  \, \scp{\Psi_{N,t}}{p_1^t\Big( N^{-2/3} \vph(x_1) -  \vphc(x_1,t) \Big) q^t_1 q^t_2 \Psi_{N,t}},
\\
\begin{split}
\frac{d}{dt} \beta^b(t) &= 2 N^{-2/3} \int d^3k \, \tilde{\eta}(k) \Im \scp{\Big(N^{-2/3} a(k) - \alpha^t(k) \Big) \Psi_{N,t}}{N e^{-ikx_1} \Psi_{N,t}} \\
&\quad - 2 N^{-2/3} \int d^3k \, \tilde{\eta}(k) \Im \scp{\Big(N^{-2/3} a(k) - \alpha^t(k) \Big) \Psi_{N,t}}{ (2 \pi)^{3/2} \mathcal{F}[\rho^t](k) \Psi_{N,t}}.
\label{eq: time derivative of beta-b preparation}
\end{split}
\end{align}
\end{lemma}

\begin{proof}
The functional $\beta^{a,1}(t)$ is time-dependent, because $\Psi_{N,t}$ and $(\varphi_1^t, \ldots, \varphi_N^t, \alpha^t)$ evolve according to \eqref{eq:Schroedinger_equation_microscopic} and \eqref{effective_eqs} respectively. The time derivative of the projector $q_m^t:=q_m^{\varphi_1^t, \ldots, \varphi_N^t}$ with $m \in \{1, \ldots,N \}$ is given by
\begin{align}
\frac{d}{dt} q_m^t = - i N^{-1/3} \Big[ h_m^t, q_m^t \Big]
\end{align}
where $h_m^t = - \Delta_m + N^{2/3} \vphc(x_m,t)$
is the effective Hamiltonian acting on the $m$-th variable.
This leads to 
\begin{align}\label{der_beta_a_1}
\frac{d}{dt} \beta^{a,1}(t) &= \frac{d}{dt} \scp{\Psi_{N,t}}{q^t_1 \Psi_{N,t}} \nonumber\\
&= i N^{-1/3} \scp{\Psi_{N,t}}{\Big[ H_N - h_1^t, q^t_1 \Big] \Psi_{N,t}} \nonumber\\
&= i N^{-1/3} \scp{\Psi_{N,t}}{\Big[ -\Delta_1 +  \vph(x_1) - h_1^t, q^t_1 \Big] \Psi_{N,t}} 
\nonumber\\
&= i N^{-1/3}  \scp{\Psi_{N,t}}{\Big[ \vph(x_1) - N^{2/3} \vphc(x_1,t), q^t_1 \Big] \Psi_{N,t}} 
\nonumber\\
&= -2 N^{1/3}  \Im \, \scp{\Psi_{N,t}}{\Big( N^{-2/3} \vph(x_1) - \vphc(x_1,t) \Big) q^t_1 \Psi_{N,t}},
\end{align}
where we used the self-adjointness of $\vph$, $\vphc$ and $q^t_1$ in the last step. Inserting $1=p^t_1+q^t_1$, we find
\begin{align}
\eqref{der_beta_a_1} &= -2 N^{1/3}  \Im \, \scp{\Psi_{N,t}}{(p_1^t+q^t_1)\Big( N^{-2/3} \vph(x_1) - \vphc(x_1,t) \Big) q^t_1 \Psi_{N,t}} \nonumber\\
&= -2 N^{1/3}  \Im \, \scp{\Psi_{N,t}}{p_1^t\Big( N^{-2/3} \vph(x_1) - \vphc(x_1,t) \Big) q^t_1 \Psi_{N,t}},
\end{align}
since the scalar product with the two $q^t_1$ projectors is real. Analogously, one derives
\begin{align}
\frac{d}{dt} \beta^{a,2}(t) &= N^{1/3} \frac{d}{dt} \scp{\Psi_{N,t}}{q^t_1 q^t_2 \Psi_{N,t}} \nonumber\\
&= -4 \Im \, \scp{\Psi_{N,t}}{p_1^t\Big( \vph(x_1) - N^{2/3} \vphc(x_1,t) \Big) q^t_1 q^t_2 \Psi_{N,t}}.
\end{align}
The time derivative of $\beta^b(t)$ is obtained by the following calculations. Note that the expressions in the calculations are all indeed well-defined, since the domain $\mathcal{D} \left( \mathcal{N} \right) \cap \mathcal{D} \left(  \mathcal{N}  H_N \right)$ is invariant under the time evolution, see Lemma~\ref{solution_th_Schr}. Then
\begin{align}
\frac{d}{dt} \beta^{b}(t) &= N^{1/3} \int d^3k \, \frac{d}{dt} \scp{ \Big( N^{-2/3} a(k) - \alpha^t(k) \Big) \Psi_{N,t}}{ \Big( N^{-2/3} a(k) - \alpha^t(k) \Big) \Psi_{N,t}}
\nonumber\\
&= -2 \int d^3k \, \Im  \scp{\Big( N^{-2/3} a(k) - \alpha^t(k) \Big) \Psi_{N,t}}{  \left[ H_N,   \Big( N^{-2/3} a(k) - \alpha^t(k) \Big)  \right] \Psi_{N,t}}
\nonumber\\
&\quad - 2 \int d^3k \, \Im  \scp{ \Big( N^{-2/3} a(k) - \alpha^t(k) \Big) \Psi_{N,t}}{ iN^{1/3}(\partial_t\alpha^t(k)) \Psi_{N,t}}.
\end{align}
For the commutator we find
\begin{align}
\Big[ H_N, \big( N^{-2/3} a(k) - \alpha^t(k) \big) \Big] &=  N^{-2/3} \delta_N \Big[ H_f,  a(k)\Big] + N^{-2/3}  \Big[\sum_{j=1}^N \vph(x_j) , a(k)\Big] \nonumber\\
&= - N^{-2/3} \Big( \delta_N \omega(k) a(k)+  \tilde{\eta}(k) \sum_{j=1}^N e^{-ikx_j} \Big) .
\end{align}
Using \eqref{effective_eqs}, it follows that
\begin{align}\label{beta_b}
\frac{d}{dt}\beta^{b}(t) &= 2 \int d^3k \, \Im \Bigg[ \scp{\Big( N^{-2/3} a(k) - \alpha^t(k) \Big)\Psi_{N,t}}{\delta_N \omega(k) \Big(  N^{-2/3} a(k) - \alpha^t(k) \Big) \Psi_{N,t}} 
\nonumber\\
&\quad + \scp{\Big( N^{-2/3} a(k) - \alpha^t(k) \Big)\Psi_{N,t}}{ \tilde{\eta}(k)N^{-2/3} \bigg( \sum_{j=1}^N e^{-ikx_j} - (2 \pi)^{3/2} \mathcal{F}\left[ \rho^t \right](k) \bigg) \Psi_{N,t}} \Bigg] \nonumber\\
&= 2 N^{-2/3}  \int d^3k \, \tilde{\eta}(k) \Im \scp{\Big( N^{-2/3} a(k) - \alpha^t(k) \Big)\Psi_{N,t}}{ N e^{-ikx_1} \Psi_{N,t}}
\nonumber \\
&\quad - 2 N^{-2/3}  \int d^3k \, \tilde{\eta}(k) \Im \scp{\Big( N^{-2/3} a(k) - \alpha^t(k) \Big)\Psi_{N,t}}{ (2 \pi)^{3/2} \mathcal{F}\left[ \rho^t \right](k) \Psi_{N,t}},
\end{align}
since the scalar product in the first line is real.
\end{proof}

Before we prove appropriate estimates for the time derivative of $\beta(t)$, let us state a technical lemma which was already proven, e.g., in \cite{petrat_2015,bach:2015}; we give a proof here for convenience. Note that this is an important point where the antisymmetry of the wave function is used.
\begin{lemma}\label{lemma:  technical estimates}
Let $A_1 = A \otimes \id_{L^2(\mathbb{R}^{3(N-1)})} \otimes \id_{\mathcal{F}_s}$ with $A: L^2({\mathbb{R}^3}) \rightarrow L^2(\mathbb{R}^3)$ trace class and
$\Psi_{N} , \Psi_{N}' \in L^2(\mathbb{R}^{3N}) \otimes \mathcal{F}_s$ antisymmetric in $x_1$ and all other electron variables except $x_{l_1}, \ldots, x_{l_j}$. Then
\begin{align}
\label{eq:  technical estimates 1}
\abs{\SCP{\Psi_{N}}{A_1 \Psi_{N}'}} &\leq (N-j)^{-1} \norm{A}_{\tr} \big\|\Psi_{N}\big\|  \big\|\Psi_{N}'\big\|.
\end{align}
\end{lemma}

\begin{proof}
In order to prove the inequality, it is convenient to use the singular value decomposition
$A = \sum_{i \in \mathbb{N}} \mu_i \ket{\chi'_i} \bra{\chi_i}$ with $(\chi'_i)_{i \in \mathbb{N}}, (\chi_i)_{i \in \mathbb{N}}$ orthonormal bases in $L^2(\mathbb{R}^3)$ and $\mu_i \geq 0 \, \forall i \in \mathbb{N}$. Using Cauchy-Schwarz, this allows us to estimate
\begin{align}
\abs{\SCP{\Psi_N}{A_1 \Psi_N'}}
&= \abs{\sum_{i \in \mathbb{N}} \mu_i \SCP{\Psi_N}{ \ket{\chi'_i} \bra{\chi_i}_1 \Psi_N'}}
\nonumber\\
&\leq \sum_{i \in \mathbb{N}} \mu_i \SCP{\Psi_N}{ \ket{\chi'_i}\bra{\chi'_i}_1 \Psi_N}^{1/2} \,
\SCP{\Psi'_N}{ \ket{\chi_i}\bra{\chi_i}_1 \Psi'_N}^{1/2}
\nonumber \\
&= (N-j)^{-1} \sum_{i \in \mathbb{N}} \mu_i 
\SCP{\Psi_N}{ \sum_{\substack{k=1 \\ k\neq l_1,\ldots,l_j}}^N \ket{\chi'_i}\bra{\chi'_i}_k \Psi_N}^{1/2} \,
\SCP{\Psi'_N}{ \sum_{\substack{l=1 \\ l\neq l_1,\ldots,l_j}}^N \ket{\chi_i}\bra{\chi_i}_l \Psi'_N}^{1/2}.
\end{align}
Note that $\sum_{k\in K} \ket{\chi_i}\bra{\chi_i}_k$ is for all $i\in \NNN$ and $K \subset \{1,\ldots,N\}$ a projector on functions antisymmetric in all $K$-variables, since
\begin{align}
\Big(\sum_{k\in K} \ket{\chi_i}\bra{\chi_i}_k \Big)^2 \Psi_N
&= \sum_{k\in K} \sum_{l \in K} \ket{\chi_i}\bra{\chi_i}_k \,  \ket{\chi_i}\bra{\chi_i}_l \Psi_N
= \sum_{k\in K} \ket{\chi_i}\bra{\chi_i}_k \Psi_N,
\end{align}
where the last step is true because the non-diagonal terms vanish due to the antisymmetry. It follows that
\begin{align}
\abs{\SCP{\Psi_N}{A_1 \Psi_N'}} 
&\leq (N-j)^{-1} \sum_{i \in \mathbb{N}} \mu_i \norm{\Psi_N} \norm{\Psi_N'}
= (N-j)^{-1} \norm{A}_{\tr}  \norm{\Psi_N} \norm{\Psi_N'}.
\end{align}
\end{proof}

\subsection{Estimate on the time derivative of $\beta^{a,1}(t)$}

\begin{lemma}\label{lemma: time derivative of beta-a-1}
Let Assumption~\ref{main_assumption} hold and let $\Psi_{N,t}$ be the solution to \eqref{eq:Schroedinger_equation_microscopic} with initial condition $\Psi_{N,0}$, and $\varphi_1^t,\ldots,\varphi_N^t,\alpha^t$ the solution to \eqref{effective_eqs} with initial condition $\varphi_1^0,\ldots,\varphi_N^0,\alpha^0$. Then there is a $C>0$ (independent of $N$, $\delta_N$, $\Lambda$, and $t$) such that for all $t>0$,
\begin{align}
\abs{\frac{d}{dt} \beta^{a,1}(t)} \leq  e^{C \Lambda^4 (1 + \|\alpha^0\|_2)(1 + t^2)} \Big( \beta(t) + N^{-1}  \Big).
\end{align}
\end{lemma} 

\begin{proof}
Using the Fourier expansion of the radiation field we write
\begin{subequations}
\begin{align}
\frac{d}{dt} \beta^{a,1}(t) &= 2 N^{1/3} \int d^3k \, \tilde{\eta}(k) \Im \scp{\Big( N^{-2/3} a(k) - \alpha^t(k) \Big)\Psi_{N,t}}{ q^t_1 e^{-ikx_1}p^t_1 \Psi_{N,t}} \label{dta1_1} \\
&\quad -2 N^{1/3} \int d^3k \, \tilde{\eta}(k) \Im \scp{\Big( N^{-2/3} a(k) - \alpha^t(k) \Big)\Psi_{N,t}}{ p^t_1 e^{-ikx_1} q^t_1 \Psi_{N,t}}.\label{dta1_2}
\end{align}
\end{subequations}
Since $\Psi_{N,t}$ is antisymmetric in the $x$ variables, we find for the first summand
\begin{align}
\abs{\eqref{dta1_1}} &\leq 2 N^{1/3}  \int d^3k \, \abs{\tilde{\eta}(k)} \, \abs{ \scp{\Big( N^{-2/3} a(k) - \alpha^t(k) \Big)\Psi_{N,t}}{N^{-1} \sum_{m=1}^N q_m e^{-ikx_m} p_m \Psi_{N,t}} }
\nonumber\\
&\leq 2 \Bigg[ N^{1/3} \int d^3k \, \norm{\Big( N^{-2/3} a(k) - \alpha^t(k) \Big)\Psi_{N,t}}^2 \Bigg]^{1/2} \nonumber\\
&\quad\times \Bigg[ N^{1/3} \int d^3k \, \abs{\tilde{\eta}(k)}^2 N^{-2} \norm{\sum_{m=1}^N q_m e^{-ikx_m} p_m \Psi_{N,t}}^2 \Bigg]^{1/2}.
\end{align}

\noindent
We now use that by Lemma~\ref{lemma:  technical estimates}, $\norm{A_1B_2\Psi_{N}}^2 \leq (N-1)^{-1} \norm{A}_{\HS}^2 \norm{B_2\Psi_N}^2$ and $\norm{A_1\Psi_{N}}^2 \leq N^{-1} \norm{A}_{\HS}^2 \norm{\Psi_N}^2$ for all antisymmetric $\Psi_N$, Hilbert-Schmidt operators $A$ and bounded operators $B$. In the end we use the semiclassical structure, i.e., Lemma~\ref{lemma: semiclassical structure}, and find
\begin{align}
&N^{-2} \norm{\sum_{m=1}^N q_m^t e^{-ikx_m} p_m^t \Psi_{N,t}}^2 \nonumber\\
&\quad= N^{-2} \bigg( N(N-1) \scp{\Psi_{N,t}}{p_1^t e^{ikx_1}q_1^t q_2^t e^{-ikx_2} p_2^t \Psi_{N,t}} + N \scp{\Psi_{N,t}}{p_1^t e^{ikx_1}q_1^t e^{-ikx_1} p_1^t \Psi_{N,t}} \bigg) 
\nonumber\\
&\quad= N^{-1}(N-1) \scp{q_1^t e^{-ikx_1} p_1^t q_2^t \Psi_{N,t}}{q_2^t e^{-ikx_2} p_2^t q_1^t \Psi_{N,t}} + N^{-1} \norm{q_1^t e^{-ikx_1} p_1^t \Psi_{N,t}}^2 
\nonumber\\
&\quad\leq N^{-1} (N-1) \norm{q_1^t e^{-ikx_1} p_1^t q_2^t \Psi_{N,t}}^2
+ N^{-1}  \norm{q_1^t e^{-ikx_1} p_1^t  \Psi_{N,t}}^2  
\nonumber \\
&\quad\leq N^{-1} \norm{q^t e^{-ikx} p^t}_{\HS}^2 \norm{q_2^t \Psi_{N,t}}^2
+ N^{-2} \norm{q^t e^{-ikx} p^t}_{\HS}^2 \norm{\Psi_{N,t}}^2
\nonumber \\
&\quad\leq N^{-1/3} C (1 + \abs{k} ) e^{C \Lambda^4 (1 + \norm{\alpha^0}_2)(1 + t^2)} \left( \beta^{a,1}(t) + N^{-1} \right).
\end{align}
Thus, 
\begin{align}
\abs{\eqref{dta1_1}}
&\leq 2 \sqrt{\beta^b(t)} 
\left(  C e^{C \Lambda^4 (1 + \norm{\alpha^0}_2)(1 + t^2)} \left( \beta^{a,1}(t) + N^{-1}   \right) \norm{\tilde{\eta} (1 + \abs{\cdot})^{1/2}}_2^2    \right)^{1/2}
\nonumber \\
&= C e^{C \Lambda^4 (1 + \norm{\alpha^0}_2)(1 + t^2)} \norm{(1 + \abs{\cdot})^{1/2} \tilde{\eta}}_2
\sqrt{\beta^b(t)} \sqrt{\beta^{a,1}(t) + N^{-1}}.
\end{align}
For \eqref{dta1_2} we can directly use Cauchy-Schwarz. We use again $\norm{A_1\Psi_{N}}^2 \leq N^{-1} \norm{A}_{\HS}^2 \norm{\Psi_N}^2$ and Lemma~\ref{lemma: semiclassical structure} in the end and find
\begin{align}
\abs{\eqref{dta1_2}}
&\leq 2 N^{1/3} \int d^3k \, \abs{\tilde{\eta}(k)} \, \abs{ \scp{ q_1^t e^{ikx_1} p_1^t \Big( N^{-2/3} a(k) - \alpha^t(k) \Big)\Psi_{N,t}}{ q^t_1 \Psi_{N,t}}}
\nonumber \\
&\leq 2 N^{-1/6} \int d^3k \, \abs{\tilde{\eta}(k)} \, \norm{q^t e^{ikx} p^t}_{\HS} \norm{\Big( N^{-2/3} a(k) - \alpha^t(k) \Big)\Psi_{N,t}}
\norm{q^t_1 \Psi_{N,t}}
\nonumber \\
&\leq C e^{C \Lambda^4 (1 + \norm{\alpha^0}_2)(1 + t^2)} \norm{(1 + \abs{\cdot})^{1/2}\tilde{\eta}}_2 \sqrt{\beta^{b}(t)} \sqrt{\beta^{a,1}(t)}.
\end{align}
To summarize, we have
\begin{align}
\abs{\frac{d}{dt} \beta^{a,1}(t)}
&\leq  C e^{C \Lambda^4 (1 + \norm{\alpha^0}_2)(1 + t^2)}  \norm{(1 + \abs{\cdot})^{1/2} \tilde{\eta}}_2 \left( \beta^{a,1}(t) + \beta^b(t) + N^{-1} \right).
\end{align}
Since $\norm{(1 + \abs{\cdot})^{1/2} \tilde{\eta}}_2 \leq C \Lambda^{3/2}$ and using for ease of notation $\abs{x} \leq \exp(\abs{x})$, this gives
\begin{align}
\abs{\frac{d}{dt} \beta^{a,1}(t)}
&\leq e^{C \Lambda^4 (1 + \norm{\alpha^0}_2)(1 + t^2)}  \left( \beta^{a,1}(t) + \beta^b(t) + N^{-1} \right).
\end{align}
\end{proof}

\subsection{Estimate on the time derivative of $\beta^{a,2}(t)$}

\begin{lemma}\label{lemma: time derivative of beta-a-2}
Let Assumption~\ref{main_assumption} hold and let $\Psi_{N,t}$ be the solution to \eqref{eq:Schroedinger_equation_microscopic} with initial condition $\Psi_{N,0}$, and $\varphi_1^t,\ldots,\varphi_N^t,\alpha^t$ the solution to \eqref{effective_eqs} with initial condition $\varphi_1^0,\ldots,\varphi_N^0,\alpha^0$. Then there is a $C>0$ (independent of $N$, $\delta_N$, $\Lambda$, and $t$) such that for all $t>0$,
\begin{align}
\abs{\frac{d}{dt} \beta^{a,2}(t)} \leq  e^{C \Lambda^4 (1 + \|\alpha^0\|_2)(1 + t^2)} \left( \beta(t) + N^{-1}  \right).
\end{align}
\end{lemma}

\begin{proof}
We write the time derivative of $\beta^{a,2}(t)$ as
\begin{align}
&\frac{d}{dt} \beta^{a,2}(t) \nonumber\\
&\quad= -4 N^{2/3} \int d^3k \, \tilde{\eta}(k) \Im \scp{\Psi_{N,t}}{\Big( N^{-2/3} a(k) - \alpha^t(k) \Big) \Big( p^t_1 e^{ikx_1}q^t_1 - q^t_1 e^{ikx_1}p^t_1 \Big) q^t_2 \Psi_{N,t}} \nonumber\\
&\quad= -4 N^{2/3} \int d^3k \, \tilde{\eta}(k) \Im \scp{q^t_2  \Psi_{N,t}}{ \Big[ p^t_1, e^{ikx_1} \Big] \Big( N^{-2/3} a(k) - \alpha^t(k) \Big) 
 \Psi_{N,t}}
\nonumber \\
&\quad= -4 N^{2/3} (N-1)^{-1} \int d^3k \, \tilde{\eta}(k) \Im 
\scp{\sum_{m=1}^N q_m^t \Psi_{N,t}}{\Big[ p^t_1, e^{ikx_1} \Big] \Big( N^{-2/3} a(k) - \alpha^t(k) \Big) 
 \Psi_{N,t}} 
 \nonumber \\
&\qquad + 4 N^{2/3} (N-1)^{-1} \int d^3k \, \tilde{\eta}(k) \Im 
\scp{\Psi_{N,t}}{q_1^t \Big[ p^t_1, e^{ikx_1} \Big] \Big( N^{-2/3} a(k) - \alpha^t(k) \Big) \Psi_{N,t}}.
\end{align}
Here, we have symmetrized the $q_2^t$ so that we can bound the time derivative appropriately in terms of $\beta^{a,2}(t)$. Note that
\begin{align}
\norm{\sum_{m=1}^N q_m^t \Psi_{N,t}}^2 &\leq N \scp{\Psi_{N,t}}{q_1^t \Psi_{N,t}} + N^2 \scp{\Psi_{N,t}}{q_1^t q_2^t \Psi_{N,t}}
\leq N \beta^{a,1}(t) + N^{5/3} \beta^{a,2}(t).
\end{align}
We can then use Lemma~\ref{lemma:  technical estimates}, as well as
\begin{equation}
\norm{q^t\left[ p^t, e^{ikx} \right]}_{\tr} \leq \norm{\left[ p^t, e^{ikx} \right]}_{\tr} \leq \norm{p^t e^{ikx} q^t}_{\tr} + \norm{p^t e^{-ikx} q^t}_{\tr}
\end{equation}
together with Lemma~\ref{lemma: semiclassical structure} and find 
\small
\begin{align}
&\abs{\frac{d}{dt} \beta^{a,2}(t)} \nonumber\\
&\quad\leq C N^{-4/3}  \int d^3k \,  \abs{\tilde{\eta}(k)} \norm{\Big[ p^t, e^{ikx} \Big]}_{\tr}
\norm{\Big( N^{-2/3} a(k) - \alpha^t(k) \Big)\Psi_{N,t}} \bigg(  \bigg\| \sum_{j=1}^N q_j^t \Psi_{N,t} \bigg\| + 1 \bigg) 
\nonumber\\
&\quad\leq C e^{C(t)} N^{-2/3} 
\int d^3k \, (1 + \abs{k}) \abs{\tilde{\eta}(k)} \norm{\Big( N^{-2/3} a(k) - \alpha^t(k) \Big)\Psi_{N,t}}
 \bigg(  \bigg\| \sum_{j=1}^N q_j^t \Psi_{N,t} \bigg\| + 1 \bigg) 
\nonumber \\
&\quad\leq C e^{C(t)} N^{-5/6} \norm{(1 + \abs{\cdot})\tilde{\eta}}_2 \sqrt{\beta^b(t)} \bigg( \sqrt{N} \sqrt{\beta^{a,1}(t)} + N^{5/6}\sqrt{\beta^{a,2}(t)} + 1 \bigg)
\nonumber \\
&\quad\leq  C e^{C(t)}  \norm{(1 + \abs{\cdot})\tilde{\eta}}_2  \bigg( \beta^b(t) + \beta^{a,1}(t) + \beta^{a,2}(t) + N^{-5/3} \bigg),
\end{align}
\normalsize
where we abbreviated $C(t) := C \Lambda^4 (1 + \norm{\alpha^0}_2)(1 + t^2)$. Since $\norm{(1 + \abs{\cdot}) \tilde{\eta}}_2 \leq C \Lambda^2$, we arrive at
\begin{align}
\abs{\frac{d}{dt}\beta^{a,2}(t)} \leq  e^{C \Lambda^4 (1 + \|\alpha^0\|_2)(1 + t^2)} \Big( \beta(t) + N^{-5/3} \Big).
\end{align}

\end{proof}

\subsection{Estimate on the time derivative of $\beta^b(t)$}
The crucial terms in the time derivative of $\beta^b(t)$ can be estimated with a diagonalization trick similar to the one used in \cite{petrat_2015}. For the following estimates we introduce the operators
\begin{align}
\label{eq: definition  of the projector over all particles}
\mbp^{\varphi} = \sum_{m=1}^N \ketbr{\varphi}_m 
= \sum_{m=1}^N p_m^{\varphi}
\quad \text{and} \quad
\mbq^{\varphi} = 1 - \mbp^{\varphi},
\end{align}
where $\varphi \in L^2(\mathbb{R}^3)$.
They have the following properties.
\begin{lemma}
\label{lemma: properties  of the projector over all particles}
The operators $\mbp^{\varphi}$ and $\mbq^{\varphi}$ as defined in \eqref{eq: definition  of the projector over all particles} are projectors on $\mathcal{H}^{(N)}$ for all $\varphi \in L^2(\RRR^3)$. Moreover, let $\chi_1, \ldots, \chi_N \in L^2(\mathbb{R}^3)$ and $\varphi_1,\ldots,\varphi_N\in L^2(\RRR^3)$ each be orthonormal, and such that $\SPAN\{ \chi_1, \ldots, \chi_N \} = \SPAN\{ \varphi_1, \ldots, \varphi_N \}$. Then
\begin{align}\label{eq: properties  of the projector over all particles1}
\Big[ \mbq^{\chi_j} , \mbq^{\chi_k} \Big] = 0 ~\forall j,k=1,\ldots,N \quad \text{and} \quad \sum_{j=1}^N \mbq^{\chi_j} &= \sum_{m=1}^N q_m^{\varphi_1,\ldots,\varphi_N}.
\end{align}
\end{lemma}
\begin{proof}
The lemma follows from a direct computation using the antisymmetry in the fermion variables.
\end{proof}

\begin{lemma}\label{lemma: time derivative of beta-b}
Let Assumption~\ref{main_assumption} hold and let $\Psi_{N,t}$ be the solution to \eqref{eq:Schroedinger_equation_microscopic} with initial condition $\Psi_{N,0}$, and $\varphi_1^t,\ldots,\varphi_N^t,\alpha^t$ the solution to \eqref{effective_eqs} with initial condition $\varphi_1^0,\ldots,\varphi_N^0,\alpha^0$. Then there is a $C>0$ (independent of $N$, $\delta_N$, $\Lambda$, and $t$) such that for all $t>0$,
\begin{align}
\abs{\frac{d}{dt} \beta^b(t)} \leq  e^{C \Lambda^4 (1 + \|\alpha^0\|_2)(1 + t^2)} \left( \beta(t) + N^{-1}  \right) .
\end{align}
\end{lemma} 

\begin{proof}

If we insert the identity $p_1^t + q_1^t = 1$ twice, \eqref{eq: time derivative of beta-b preparation} can be written as
\begin{align}
\frac{d}{dt} \beta^b(t) = pp{-}\Term + qp{-}\Term + pq{-}\Term + {qq}{-}\Term
\end{align}
with
\begin{align}
\label{pp}pp{-}\Term &=  2 N^{-2/3}  \int d^3k \, \tilde{\eta}(k) \Im \scp{\Big(  N^{-2/3} a(k) - \alpha^t(k) \Big)\Psi_{N,t}}{\sum_{i=1}^N p_i^t e^{-ikx_i}p_i^t \Psi_{N,t}}
\nonumber\\
& \quad - 2 N^{-2/3}  \int d^3k \, \tilde{\eta}(k) \Im \scp{\Big(  N^{-2/3} a(k) - \alpha^t(k) \Big)\Psi_{N,t}}{ \int d^3y\, e^{-iky} \rho^t(y)  \Psi_{N,t}},
\\
qp{-}\Term &= 2 N^{1/3} \int d^3k \, \tilde{\eta}(k) \Im \scp{\Big( N^{-2/3} a(k) - \alpha^t(k) \Big)\Psi_{N,t}}{ q_1^t e^{-ikx_1}p_1^t \Psi_{N,t}}, 
\\
pq{-}\Term &= 2 N^{1/3} \int d^3k \, \tilde{\eta}(k) \Im \scp{\Big(  N^{-2/3} a(k) - \alpha^t(k) \Big)\Psi_{N,t}}{ p_1^t e^{-ikx_1}q_1^t \Psi_{N,t}}, 
\\
qq{-}\Term &= 2 N^{1/3} \int d^3k \, \tilde{\eta}(k) \Im \scp{\Big(  N^{-2/3} a(k) - \alpha^t(k) \Big)\Psi_{N,t}}{ q_1^t e^{-ikx_1}q_1^t \Psi_{N,t}}.
\end{align}
To estimate the $pp{-}\Term$, we split $e^{-ikx} = \cos(kx) - i \sin(kx)$ into its real and imaginary part. Subsequently we only estimate the $\cos$-terms $pp{-}\Term_{\cos}$; the $\sin$-terms are estimated in exactly the same manner. Note that for each fixed $t > 0$ and $k\in\RRR^3$, we can regard $p^t\cos(kx)p^t$ as a Hermitian $N\times N$ matrix on $\SPAN\{\varphi_1^t,\ldots,\varphi_N^t\}$. By the spectral theorem we can find orthonormal $\chi_1^{t,k}, \ldots, \chi_N^{t,k} \in \SPAN\{\varphi_1^t,\ldots,\varphi_N^t\}$ (i.e., $p^t = \sum_{j=1}^N \ketbr{\varphi_j^t} = \sum_{j=1}^N \ketbr{\chi_j^{t,k}}$) and real $\lambda_1^{t,k},\ldots,\lambda_N^{t,k}$ such that $p^t \cos(kx_1)p^t = \sum_{j=1}^N \lambda_j^{t,k} \ketbr{\chi_j^{t,k}}$. In particular this implies
\begin{align}
\abs{\lambda_j^{t,k}} &= \abs{\scp{\chi_j^{t,k}}{\cos(kx)\chi_j^{t,k}}} \leq 1, \label{diag_prop_1}
\\
\sum_{j=1}^N \lambda_j^{t,k} &= \tr \Big(p^t\cos(kx)p^t\Big) = \int d^3y \, \cos(ky) \rho^t(y), \label{diag_prop_2}
\\
\sum_{i=1}^N p_i^t \cos(k x_i) p_i^t &= \sum_{i=1}^N \sum_{j=1}^N \lambda_j^{t,k} \ketbr{\chi_j^{t,k}}_i = \sum_{j=1}^N \lambda_j^{t,k} \mbp^{\chi_j^{t,k}}. \label{diag_prop_3}
\end{align}
Using \eqref{diag_prop_2} and \eqref{diag_prop_3}, the $\cos$-part of the $pp{-}\Term$ can be written as
\begin{align}
pp{-}\Term_{\cos} &=  2 N^{-2/3}  \int d^3k \, \tilde{\eta}(k)  \Im \scp{\Big( N^{-2/3} a(k) - \alpha^t(k) \Big)\Psi_{N,t}}{ \sum_{j=1}^N \lambda_j^{t,k} \Big( \mbp^{\chi_j^{t,k}} - 1 \Big) \Psi_{N,t}} 
\nonumber \\
&= 2 N^{-2/3}  \int d^3k \, \tilde{\eta}(k)  \Im \scp{\Big( N^{-2/3} a(k) - \alpha^t(k) \Big)\Psi_{N,t}}{ \sum_{j=1}^N \lambda_j^{t,k} \mbq^{\chi_j^{t,k}}  \Psi_{N,t}}
\end{align}
and be estimated by
\begin{align}
\label{eq: beta-b-pp term}
\abs{pp{-}\Term_{\cos}}
&\leq 2 N^{-2/3} \int d^3k \,   \norm{\Big( N^{-2/3} a(k) - \alpha^t(k) \Big)\Psi_{N,t}}  \abs{\tilde{\eta}(k)}
\bigg\| \sum_{j=1}^N \lambda_j^{t,k}  \mbq^{\chi_j^{t,k}}  \Psi_{N,t} \bigg\|
\nonumber \\
&\leq 2 \left(  \int d^3k \, N^{1/3}  \norm{\Big( N^{-2/3} a(k) - \alpha^t(k) \Big)\Psi_{N,t}}^2  \right)^{1/2}
\nonumber \\
&\quad \times 
\left(  \int d^3k \,  N^{-5/3}  \abs{\tilde{\eta}(k)}^2  
\bigg\| \sum_{j=1}^N \lambda_j^{t,k} \mbq^{\chi_j^{t,k}}  \Psi_{N,t} \bigg\|^2 \right)^{1/2}
\nonumber \\
&= 2 \sqrt{\beta^b(t)}   \left(  \int d^3k \,  N^{-5/3}  \abs{\tilde{\eta}(k)}^2  
\bigg\| \sum_{j=1}^N \lambda_j^{t,k} \mbq^{\chi_j^{t,k}}  \Psi_{N,t} \bigg\|^2 \right)^{1/2}.
\end{align}
If one makes use of \eqref{diag_prop_1} and Lemma \ref{lemma: properties  of the projector over all particles} one finds
\begin{align}
\bigg\| \sum_{j=1}^N \lambda_j^{t,k} \mbq^{\chi_j^{t,k}}  \Psi_{N,t} \bigg\|^2
&= \sum_{i,j=1}^N \overline{\lambda_i^{t,k}} \lambda_j^{t,k} \scp{\mbq^{\chi_i^{t,k}} \Psi_{N,t}}{\mbq^{\chi_j^{t,k}} \Psi_{N,t}}
\nonumber \\
&\leq \sum_{i,j=1}^N \abs{ \scp{\mbq^{\chi_i^{t,k}} \Psi_{N,t}}{\mbq^{\chi_j^{t,k}} \Psi_{N,t}}} \nonumber\\
&= \sum_{i,j=1}^N \scp{\mbq^{\chi_i^{t,k}} \Psi_{N,t}}{\mbq^{\chi_j^{t,k}} \Psi_{N,t}}
\nonumber \\
&= \scp{\Psi_{N,t}}{\sum_{i=1}^N q_i^t \sum_{j=1}^N q_j^t \Psi_{N,t}} \nonumber\\
&= N(N-1) \scp{\Psi_{N,t}}{q_1^t q_2^t \Psi_{N,t}} + N \scp{\Psi_{N,t}}{q_1^t \Psi_{N,t}}
\nonumber \\
&\leq N \beta^{a,1}(t) + N^{5/3} \beta^{a,2}(t)
\end{align}
and obtains 
\begin{align}
\abs{pp{-}\Term_{\cos}}
&\leq 2 \norm{\tilde{\eta}}_2 \sqrt{\beta^b(t)} \sqrt{\beta^{a,1}(t) + \beta^{a,2}(t)} 
\leq C \Lambda \left(  \beta^{a,1}(t) + \beta^{a,2}(t) + \beta^b(t) \right).
\end{align}
In exactly the same manner one estimates  $pp{-}\Term_{\sin}$ and obtains $\abs{pp{-}\Term} \leq  C \Lambda \beta(t)$.
From the observation that $qp{-}\Term = \eqref{dta1_1}$ and $pq{-}\Term = - \eqref{dta1_2}$ we immediately get
\begin{align}
\abs{qp{-}\Term + pq{-}\Term}
&\leq e^{C \Lambda^4 (1 + \norm{\alpha^0}_2)(1 + t^2)}  \left( \beta^{a,1}(t) + \beta^b(t) + N^{-1} \right).
\end{align}
Similar to \eqref{eq: beta-b-pp term} we estimate
\begin{align}
\abs{qq{-}\Term}
&= 2 N^{-2/3} \abs{\int d^3k \, \tilde{\eta}(k) \scp{\Big(  N^{-2/3} a(k) - \alpha^t(k) \Big)\Psi_{N,t}}{\sum_{m=1} q_m^t e^{-ikx_m}q_m^t \Psi_{N,t}}}
\nonumber \\
&\leq 2 \sqrt{\beta^b(t)} 
\left( \int d^3k \,  N^{-5/3} \abs{\tilde{\eta}(k)}^2 \bigg\| \sum_{m=1}^N q_m^t e^{-ikx_m} q_m^t  \Psi_{N,t} \bigg\|^2  \right)^{1/2}.
\end{align}
By means of 
\begin{align}
&\bigg\| \sum_{m=1}^N q_m^t e^{-ikx_m} q_m^t \Psi_{N,t} \bigg\|^2 \nonumber\\
&\quad= N (N-1) \scp{ q_1^t e^{-ikx_1} q_1^t  \Psi_{N,t}}{ q_2^t e^{-ikx_2} q_2^t  \Psi_{N,t}} 
+ N \norm{ q_1^t e^{-ikx_1} q_1^t  \Psi_{N,t}}^2
\nonumber \\
&\quad\leq N^2 \scp{ e^{-ikx_1} q_1^t q_2^t \Psi_{N,t}}{ e^{-ikx_2} q_1^t q_2^t  \Psi_{N,t}} 
+ N \norm{ q_1^t e^{-ikx_1} q_1^t  \Psi_{N,t}}^2
\nonumber \\
&\quad\leq N^2 \norm{q_1^t q_2^t \Psi_{N,t}}^2 + N \norm{ q_1^t  \Psi_{N,t}}^2 \nonumber\\
&\quad= N \beta^{a,1}(t) + N^{5/3} \beta^{a,2}(t)
\end{align}
this becomes
\begin{align}
\abs{qq{-}\Term}
&\leq 2 \norm{\tilde{\eta}}_2 \sqrt{\beta^b(t)} \sqrt{\beta^{a,1}(t) + \beta^{a,2}(t)} 
\leq C \Lambda \beta(t).
\end{align}
Summing all terms up then shows Lemma \ref{lemma: time derivative of beta-b}.
\end{proof}

\subsection{The Gronwall estimate}

\begin{lemma}\label{lemma: time derivative of beta}
Let Assumption~\ref{main_assumption} hold and let $\Psi_{N,t}$ be the solution to \eqref{eq:Schroedinger_equation_microscopic} with initial condition $\Psi_{N,0}$, and $\varphi_1^t,\ldots,\varphi_N^t,\alpha^t$ the solution to \eqref{effective_eqs} with initial condition $\varphi_1^0,\ldots,\varphi_N^0,\alpha^0$. Then there is a $C>0$ (independent of $N$, $\delta_N$, $\Lambda$, and $t$) such that for all $t>0$,
\begin{align}
\label{eq: Gronwall for beta}
\beta(t) \leq  e^{e^{C \Lambda^4 (1 + \|\alpha^0\|_2)(1 + t^2)}} \left(\beta(0) + N^{-1}  \right).
\end{align}
\end{lemma} 

\begin{proof}
If we use Lemmas \ref{lemma: time derivative of beta-a-1}, \ref{lemma: time derivative of beta-a-2} and \ref{lemma: time derivative of beta-b} we get
\begin{align}
\frac{d}{dt} \beta(t)
&\leq \abs{\frac{d}{dt} \beta^{a,1}(t)} + \abs{\frac{d}{dt} \beta^{a,2}(t)} + \abs{\frac{d}{dt} \beta^b(t)}
\leq e^{C \Lambda^4 (1 + \|\alpha^0\|_2)(1 + t^2)} \left( \beta(t) + N^{-1}  \right).
\end{align}
Applying Gronwall's Lemma we obtain
\begin{align}
\beta(t) &\leq e^{\int_0^t ds \, e^{C \Lambda^4 (1 + \|\alpha^0\|_2)(1 + s^2)}} \beta(0) 
+ \Big( e^{\int_0^t ds \, e^{ C \Lambda^4 (1 + \|\alpha^0\|_2)(1 + s^2)}}
- 1 \Big) N^{-1}
\nonumber \\
&\leq e^{\int_0^t ds \, e^{C \Lambda^4 (1 + \|\alpha^0\|_2)(1 + s^2)}} \left( \beta(0) + N^{-1} \right).
\end{align}
Using $\int_0^t ds \; e^{C \Lambda^4 (1 + \|\alpha^0\|_2)(1 + s^2)} \leq e^{\tilde{C} \Lambda^4 (1 + \|\alpha^0\|_2)(1 + t^2)}$ for some $\tilde{C}>0$ shows the claim.
\end{proof}

\section{Proof of Theorem~\ref{theorem: main theorem}}\label{sec_main_proof}

In order to state our main result in terms of the trace norm difference of reduced density matrices let us add the following lemma.
\begin{lemma}\label{lemma:beta_red_den}
Let $\varphi_1, \ldots, \varphi_N \in L^2(\mathbb{R}^3)$ be orthonormal,  $\alpha \in L^2(\mathbb{R}^3)$ and $\Psi_{N} \in \mathcal{H}^{(N)} \cap \mathcal{D}\left(\mathcal{N} \right)$ with $\norm{\Psi_N} =1$. Then
\begin{align}
\label{eq:betaa_1_and_tr}
2 \beta^{a,1}(\Psi_{N}, \varphi_1, \ldots, \varphi_N) &\leq \norm{\gamma^{(1,0)}_{N} - N^{-1} p}_{\tr}
\leq  \sqrt{8 \beta^{a,1}(\Psi_N, \varphi_1, \ldots, \varphi_N)},
\\
\label{eq:betab_and_tr}
\norm{\gamma_{N}^{(0,1)} - \ket{\alpha}\bra{\alpha}}_{\tr}  &\leq  3 N^{-1/3} \beta^b(\Psi_N,\alpha) + 6 \norm{\alpha}_2 \sqrt{ N^{-1/3} \beta^b(\Psi_N,\alpha)}. 
\end{align}
\end{lemma}

\begin{proof}
This is a standard result. For example, a proof of \eqref{eq:betaa_1_and_tr} can be found in \cite[Section~3.1]{petrat_2015} and a proof of \eqref{eq:betab_and_tr} in \cite[Section~VII]{leopold3}.
\end{proof}

Let us now summarize all estimates and put them together for a proof of our main result.

\begin{proof}[Proof of Theorem~\ref{theorem: main theorem}]
Let us first note that from Lemma~\ref{solution_th_Schr} we have that $\Psi_{N,t} \in \mathcal{H}^{(N)} \cap \mathcal{D}(\mathcal{N}) \cap \mathcal{D}(\mathcal{N}H_N)$ for all $t\geq 0$, and from Proposition~\ref{theorem: solution theory} that $(\varphi_1^t, \ldots, \varphi_N^t, \alpha^t) \in H^2(\mathbb{R}^3)^N \times L_1^2(\mathbb{R}^3)$ for all $t\geq 0$. From Lemma~\ref{lemma: time derivative of beta} we obtain the Gronwall estimate
\begin{align}\label{eq: Gronwall for beta again}
\beta(t) \leq  e^{e^{C \Lambda^4 (1 + \|\alpha^0\|_2)(1 + t^2)}} \left(\beta(0) + N^{-1}  \right).
\end{align}
Recall that $\beta = \beta^{a,1} + \beta^{a,2} + \beta^b$. From the first inequality of Lemma~\ref{lemma:beta_red_den} and from \eqref{beta_b_c} we get
\begin{equation}
\beta(0) \leq a_N + b_N + c_N,
\end{equation}
so that
\begin{align}\label{eq: Gronwall for beta again_and_again}
\beta(t) \leq  e^{e^{C \Lambda^4 (1 + \|\alpha^0\|_2)(1 + t^2)}} I_N,
\end{align}
where we abbreviated $I_N=a_N + b_N + c_N + N^{-1}$. Since $\beta^{a,1}$, $\beta^{a,2}$ and $\beta^b$ are all positive we then get with Lemma~\ref{lemma:beta_red_den} that
\begin{align}
\norm{\gamma_{N,t}^{(1,0)} - N^{-1} p^t}_{\tr} \leq \sqrt{8\beta(t)} \leq e^{e^{C \Lambda^4 (1 + \|\alpha^0\|_2)(1 + t^2)}} \sqrt{I_N}
\end{align}
for some $C>0$. From Lemma~\ref{lemma: growth of alpha-t} we know that $\norm{\alpha^t}_2 \leq \norm{\alpha^0}_2 +   \norm{\tilde{\eta}}_2 \abs{t}$ and thus 
\begin{align}
\norm{\gamma_{N,t}^{(0,1)} - \ket{\alpha^t}\bra{\alpha^t}}_{\tr} &\leq 3 N^{-1/3} \beta^b(t) + 6 \norm{\alpha^t}_2 \sqrt{ N^{-1/3} \beta^b(t)} \nonumber\\
&\leq e^{e^{C \Lambda^4 (1 + \|\alpha^0\|_2)(1 + t^2)}} \left( N^{-1/3}I_N + \sqrt{N^{-1/3}I_N} \right),
\end{align}
which gives \eqref{eq: Nelson main theorem 2} for some $C > 0$, if $c_N \leq \tilde{C}N^{1/3}$ for some $\tilde{C}>0$ is assumed.

In the theorem we also provide bounds for the specific initial state $\bigwedge_{j=1}^N \varphi_j^0 \otimes W(N^{2/3} \alpha^0) \Omega$. Since for this state $\gamma_{N,0}^{(1,0)} = N^{-1}p^0$ we have $a_N=0$, $b_N=0$ because $q_1 \bigwedge_{j=1}^N \varphi_j^0 \otimes W(N^{2/3} \alpha^0) \Omega = 0$,  and also
\begin{equation}
c_N = N^{-1} \SCP{ W^{-1}(N^{2/3} \alpha^0) W(N^{2/3} \alpha^0) \Omega}{ \mathcal{N} W^{-1}(N^{2/3} \alpha^0) W(N^{2/3} \alpha^0) \Omega} = N^{-1} \SCP{\Omega}{\mathcal{N} \Omega} = 0.
\end{equation}
Furthermore, we have
\begin{align}
\bigwedge_{j=1}^N \varphi_j^0 \otimes W(N^{2/3} \alpha^0) \Omega \in \left( L_{as}^2(\mathbb{R}^{3N})  \otimes \mathcal{F}_s \right) \cap \mathcal{D} \left( \mathcal{N} \right) \cap  \mathcal{D} \left(  \mathcal{N}  H_N \right),
\end{align}
which can be checked by direct calculation as in \cite[Section~IX]{leopold}.
\end{proof}

\appendix

\section{Appendix: Convergence to the Free Evolution}\label{section:convergence_free evolution}
In this section we prove Theorem \ref{theorem: free evolution}.
\begin{proof}[Proof of Theorem \ref{theorem: free evolution}]
We recall that $h^0 = - \Delta + N^{2/3} \vphc(\cdot,0)$ and define a family of unitary operators by $i N^{1/3} \partial_t U(t) = h^0 U(t)$ and $U(0) =  \id$, i.e., such that
\begin{align}
i N^{1/3} \partial_t \big( U^*(t) \tilde{p}^t U(t) \big)
&= - U^*(t) h^0 \tilde{p}^t U(t) 
+ U^*(t) \Big[ h^0, \tilde{p}^t \Big] U(t) 
+ U^*(t) \tilde{p}^t h^0 U(t)
= 0.
\end{align}
Similarly, we obtain
\begin{align}
i N^{1/3} \partial_t \big( U^*(t) p^t U(t) \big) = U^*(t) \Big[ h^t - h^0 , p^t \Big] U(t) = N^{2/3}  U^*(t)  \Big[ \vphc(\cdot,t) - \vphc(\cdot,0), p^t  \Big] U(t).
\end{align}
With Duhamel's formula we conclude
\begin{align}
p^t 
&= U(t) p^0 U^*(t) 
- i N^{1/3} \int_0^t ds \, U(t-s) \Big[ \vphc(\cdot,s) - \vphc(\cdot,0) , p^s \Big] U(s-t).
\end{align}
Thus, if we use that $\tilde{p}^t = U(t) \tilde{p}^0 U^*(t) = U(t) p^0 U^*(t)$ we get
\begin{align}
\norm{ p^t - \tilde{p}^t}_{\tr} 
&\leq N^{1/3} \norm{ \int_0^t ds \, U(t-s) \Big[ \vphc(\cdot,s) - \vphc(\cdot,0) , p^s  \Big] U(s-t)}_{\tr}
\nonumber \\
&\leq N^{1/3} \int_0^t ds \,   \norm{ \Big[ \vphc(\cdot,s) - \vphc(\cdot,0) , p^s  \Big] }_{\tr}.
\end{align}
By means of the Duhamel expansion of Equation~\eqref{effective_eqs} for $\alpha^s$,
\begin{align}
\alpha^s(k)
&= e^{- i N^{-1/3} \delta_N \omega(k) s} \alpha^0(k)
- i N^{-1} (2 \pi)^{3/2} \tilde{\eta}(k) \int_0^s du \, e^{- i N^{-1/3} \delta_N \omega(k) (s-u)} 
\mathcal{F}[\rho^u](k),
\end{align}
and $\overline{\mathcal{F}[\rho^u](k)} = \mathcal{F}[\rho^u](-k)$ one obtains
\begin{align}
\vphc(x,s) 
&= \int d^3k \, \tilde{\eta}(k) \big( e^{ikx} e^{- i N^{-1/3} \delta_N \omega(k) s} \alpha^0(k)
+ e^{-ikx} e^{i N^{-1/3} \delta_N \omega(k) s} \overline{\alpha^0(k)} \big) 
\nonumber \\
&\quad 
- 2 N^{-1} (2 \pi)^{3/2} \int d^3k \, \abs{\tilde{\eta}(k)}^2 e^{ikx} \int_0^s du \, \sin \big( N^{-1/3} \delta_N \omega(k) (s-u) \big) \mathcal{F}[\rho^u](k).
\end{align}
We continue the previous inequality and get
\begin{subequations}
\begin{align}
\label{eq: free evolution derivation 2}
&\norm{ p^t - \tilde{p}^t}_{\tr} \leq 2
N^{1/3} \int_0^t ds \norm{\int d^3k \, \tilde{\eta}(k) \big( e^{i N^{-1/3} \delta_N \omega (k) s} - 1 \big)
\overline{\alpha^0(k)} \big[ e^{-ikx} , p^s  \big]   }_{\tr}
\\
\label{eq: free evolution derivation 3}
&\quad +
C N^{-2/3} \int_0^t ds \, \norm{\int d^3k \, \abs{\tilde{\eta}(k)}^2 \int_0^s du \, \sin \big( N^{-1/3} \delta_N \omega(k) (s-u) \big) \mathcal{F}[\rho^u](k) \big[ e^{ikx}, p^s  \big] }_{\tr}  .
\end{align}
\end{subequations}
In the following, we use \eqref{eq: semiclassical structure},  $\abs{e^{ix} -1} \leq 2 \abs{x}$ and $\abs{x} \leq e^{\abs{x}}$ to bound the first line by
\begin{align}
\eqref{eq: free evolution derivation 2}
&\leq 2 N^{1/3} \int_0^t ds \int d^3k \, \tilde{\eta}(k)  \abs{e^{ i N^{-1/3} \delta_N \omega(k) s} -1} 
\abs{\alpha^0(k)}  \norm{ \big[ e^{ - ikx}, p^s \big] }_{\tr}
\nonumber \\
&\leq 4 \delta_N \int_0^t ds \int d^3k \, \abs{s} \omega(k) \tilde{\eta}(k)  
\abs{\alpha^0(k)}  \norm{ \big[ e^{ - ikx}, p^s \big] }_{\tr}
\nonumber \\
&\leq C N^{2/3} \delta_N \int_0^t ds\, \abs{s} e^{C(s)} \int d^3k \, (1 + \abs{k}) \omega(k) \tilde{\eta}(k) \abs{\alpha^0(k)}
\nonumber \\
&\leq C N^{2/3} \delta_N (1 + \Lambda)^3 \norm{\alpha_0}_2 \int_0^t ds \, e^{C(s)},
\end{align}
where 
$C(s) = C \Lambda^4 (1 + \norm{\alpha^0}_2) ( 1 + s^2)$. Then, we notice that 
$\abs{\mathcal{F}[\rho^u](k)} \leq (2 \pi)^{-3/2} \norm{\rho^u}_1 = (2 \pi)^{-3/2} N$ and use $\abs{\sin(x)} \leq \abs{x}$ to estimate 
\begin{align}
\eqref{eq: free evolution derivation 3}
&\leq C N^{-2/3} \int_0^t ds \int d^3k  \int_0^s du \, \abs{\tilde{\eta}(k)}^2 \abs{\sin \big( N^{-1/3} \delta_N \omega(k) (s-u) \big)} 
\abs{\mathcal{F}[\rho^u](k)} \norm{\big[ e^{ikx} , p^s  \big]}_{\tr}
\nonumber \\
&\leq C \delta_N \int_0^t ds \int d^3k  \int_0^s du \,  \abs{s-u}  \omega(k)  \abs{\tilde{\eta}(k)}^2
\norm{\big[ e^{ikx} , p^s  \big]}_{\tr}
\nonumber \\
&\leq C N^{2/3} \delta_N  \int_0^t ds \int d^3k \int_0^s du \, \abs{s-u} e^{C(s)} \id_{\abs{k} \leq \Lambda}(k) (1 + \abs{k} )
\nonumber \\
&\leq C N^{2/3} \delta_N (1 + \Lambda )^4 \int_0^t ds \,  e^{C(s)} .
\end{align}
Collecting the estimates and using $C (1 + \Lambda)^4  \int_0^t ds \, e^{ C(s)} \leq e^{ C(t)}$ proves \eqref{eq:free_dyn_to_Hartree}. Then \eqref{eq:free_dyn_to_Schr} follows from \eqref{eq: Nelson main theorem 1} and the triangle inequality.
\end{proof}

\section{Appendix: The Fermionic Schr\"odinger-Klein-Gordon Equations}
\label{section: Appendix B: The fermionic Schr\"odinger-Klein-Gordon equations}

Subsequently, we prove the existence and uniqueness of solutions of the effective equations \eqref{effective_eqs}. To shorten the notation we introduce
\begin{align}
\begin{split}
\Hilbert &=  \bigoplus_{n=1}^{N+1} L^2(\mathbb{R}^3), \quad  
\vec{\varphi} = (\varphi_1, \ldots, \varphi_N),
\quad \rho_{\vec{\varphi}}(x) = \sum_{j=1}^N \abs{\varphi_j(x)}^2, \\
\vphc^{\alpha}(x) &= \int d^3k \, \tilde{\eta}(k) \big(e^{ikx} \alpha(k) + e^{-ikx} \overline{\alpha(k)} \big).
\end{split}
\end{align}
Then, we define the operator $A: \mathcal{D}(A) \rightarrow \Hilbert$ as the orthogonal sum
\begin{align}
A = A_1 \oplus A_2 \oplus \ldots \oplus A_{N+1}, ~\text{with}~ A_j = \begin{cases}
1 - N^{-1/3} \Delta &\text{for}~ j \in \{1, 2, \ldots, N \}, \\
1 + N^{-1/3} \delta_N \omega &\text{for}~ j=N+1.
\end{cases}
\end{align}
Moreover, we define $J: \mathcal{D}(A) \rightarrow \Hilbert$ by
\begin{align}
J_j [(\varphi_1, \ldots, \varphi_N, \alpha)]  = 
\begin{cases}
- i (N^{1/3} \vphc^{\alpha} - 1)  \varphi_j \quad &\text{if} \; j \in \{1, \ldots, N \}, \\
- i (- \alpha + N^{-1} (2 \pi)^{3/2} \tilde{\eta} \mathcal{F}[\rho_{\vec{\varphi}}])   \quad  &\text{if} \;  j= N+1.
\end{cases}
\end{align}
The fermionic Schr\"odinger-Klein-Gordon equations \eqref{effective_eqs} can then be written as
\begin{align}
\label{eq: solution theory effective_eqs}
\frac{d}{dt} \begin{pmatrix}
\vec{\varphi}^{\,t} \\  \alpha^t
\end{pmatrix}
& = - i A 
\begin{pmatrix}
\vec{\varphi}^{\,t} \\  \alpha^t
\end{pmatrix}
+  J [(\vec{\varphi}^{\,t}, \alpha^t)].
\end{align}
The goal of this section is to show
\begin{lemma}
\label{lemma: solution theory FSKGE}
Let $N \in \mathbb{N} \setminus \{0 \}, \Lambda \in [1,\infty)$, $m\in[0,\infty)$, and $\delta_N \in (0, \infty)$. Then
\begin{itemize}
\item[a)]
$A$ is a self-adjoint operator on $\Hilbert$ with $\mathcal{D}(A) = \big( \bigoplus_{n=1}^{N} H^2(\mathbb{R}^3) \big) \oplus L_1^2(\mathbb{R}^3)$,

\item[b)] $J$ is a mapping which takes $\mathcal{D}(A)$ into $\mathcal{D}(A)$,

\item[c)]
$\|J[(\vec{\varphi},\alpha)] - J[(\vec{\psi},\beta)] \|_{\Hilbert} \leq C_{N,\Lambda }\big( \|(\vec{\varphi}, \alpha)\|_{\Hilbert}, \|(\vec{\psi}, \beta)\|_{\Hilbert} \big) \, \|(\vec{\varphi}, \alpha)  -  (\vec{\psi}, \beta)  \|_{\Hilbert}$,

\item[d)]
$\|A J[(\vec{\varphi},\alpha)] \|_{\Hilbert} \leq C_{N,\Lambda,m} \big( \|(\vec{\varphi}, \alpha)\|_{\Hilbert}\big) \, 
\| A (\vec{\varphi}, \alpha) \|_{\Hilbert}$,

\item[e)]
$\| A J[(\vec{\varphi},\alpha)] - A J[(\vec{\psi},\beta)] \|_{\Hilbert} \vspace{2mm} \\ {}\quad \leq C_{N,\Lambda,m, \delta_N }\big( \|(\vec{\varphi}, \alpha)\|_{\Hilbert}, \|(\vec{\psi}, \beta)\|_{\Hilbert}, \| A(\vec{\psi}, \beta)\|_{\Hilbert} \big) \| A(\vec{\varphi}, \alpha)  -  A (\vec{\psi}, \beta) \|_{\Hilbert}$,
\end{itemize}
where each $C$ is a monotone increasing (everywhere finite) function of all its variables. Moreover, let $(\vec{\varphi}^{\,0},\alpha^0) \in \mathcal{D}(A)$ 
and assume there is a $T >0$ so that \eqref{eq: solution theory effective_eqs} has a unique continuously differentiable solution for $t \in [0,T)$. Then, $\|(\vec{\varphi}^{\,t},\alpha^t)\|_{\Hilbert}$ is bounded from above for all $t \in [0,T)$.
\end{lemma}
We give the proof of the lemma below. In order to prove Proposition~\ref{theorem: solution theory} we use \cite[Theorem~X.74]{reedsimon} with $n=1$.

\begin{proof}[Proof of Proposition~\ref{theorem: solution theory}]
From the statements in parts a)--e) in Lemma~\ref{lemma: solution theory FSKGE} we have that all conditions of part (a) of Theorem X.73 in \cite{reedsimon} are satisfied. This implies the existence of a $T >0$ and a unique continuously differentiable solution to \eqref{eq: solution theory effective_eqs} for $t\in[0,T)$ and for all $(\vec{\varphi}^{\, 0},\alpha^0) \in \mathcal{D}(A)$.
By the second part of Lemma~\ref{lemma: solution theory FSKGE} this solution is bounded in norm for all $t\in[0,T)$. Proposition~\ref{theorem: solution theory} then follows from Lemma~\ref{lemma: solution theory FSKGE} and \cite[Theorem~X.74]{reedsimon} with $n=1$. The $\varphi_1^t,\ldots,\varphi_N^t$ are orthonormal for all $t\in[0,\infty)$ because $ - \Delta + N^{2/3} \vphc(\cdot,t)$ is a symmetric operator.
\end{proof}

Before we prove Lemma \ref{lemma: solution theory FSKGE}, let us show that on the chosen time scale $\norm{\alpha^t}_2$ remains of order one during the time evolution.

\begin{lemma}
\label{lemma: growth of alpha-t} 
Let $(\varphi_1^t, \ldots, \varphi_N^t, \alpha^t)$ be the solution to \eqref{effective_eqs} with $(\varphi_1^0, \ldots, \varphi_N^0,\alpha^0) \in  H^2(\mathbb{R}^3) \oplus L_1^2(\mathbb{R}^3)$ and orthonormal $\varphi_1^0, \ldots, \varphi_N^0$, and let $\tilde{\eta}$ be defined as in \eqref{eq: Nelson cut off function}. Then
\begin{align}
\norm{\alpha^t}_2 \leq \norm{\alpha^0}_2 +   \norm{\tilde{\eta}}_2 \abs{t} .
\end{align}
\end{lemma}
\begin{proof}
We define $U_{\omega}(t) = e^{-iN^{-1/3}\delta_N\omega(k)t}$. Then the Duhamel expansion of Equation~\eqref{effective_eqs} for $\alpha^t$ can be written as
\begin{align}
\alpha^t = U_{\omega}(t)\alpha^0 - i \int_0^t ds \, U_{\omega}(t-s) N^{-1} (2\pi)^{3/2}  \, \tilde{\eta} \, \mathcal{F}[\rho_{\vec{\varphi}^{\,s}}].
\end{align}
Then, since $\norm{\mathcal{F}[\rho_{\vec{\varphi}^{\,s}}]}_{\infty} \leq (2 \pi)^{-3/2} \norm{\rho_{\vec{\varphi}^{\,s}}}_1 = (2 \pi)^{-3/2} N$ for all $s\in \RRR$, we have
\begin{align}
\norm{\alpha^t}_2 &\leq \norm{U_{\omega}(t)\alpha^0}_2 + N^{-1} (2\pi)^{3/2}  \int_0^t ds\, \norm{U_{\omega}(t-s) \, \tilde{\eta} \mathcal{F}[\rho_{\vec{\varphi}^{\,s}}]}_2 \nonumber\\
&\leq \norm{\alpha^0}_2 + N^{-1} (2\pi)^{3/2}  \norm{\tilde{\eta}}_2 \int_0^t ds\, \norm{\mathcal{F}[\rho_{\vec{\varphi}^{\,s}}]}_\infty \nonumber\\
&\leq \norm{\alpha^0}_2 + \norm{\tilde{\eta}}_2 \abs{t} .
\end{align}
\end{proof}

\begin{proof}[Proof of Lemma~\ref{lemma: solution theory FSKGE}]{~}
\subsubsection*{Part a)}
The operators $A_j = (1 - N^{-1/3} \Delta)$ with $\mathcal{D}(A_j) = H^2(\mathbb{R}^3)$ and $j \in \{1, 2, \ldots, N \}$
as well as  $A_{N+1} = 1 + N^{-1/3} \delta_N \omega$ with  $\mathcal{D}(A_{N+1}) =\{ \alpha \in L^2(\mathbb{R}^3) | A_{N+1} \alpha \in L^2(\mathbb{R}^3) \} = L_1^2(\mathbb{R}^3)$ are clearly self-adjoint. The fact that direct sums of self-adjoint operators are self-adjoint (see, e.g., \cite[Theorem 2.24]{teschl}) proves part a) of Lemma~\ref{lemma: solution theory FSKGE}.

\subsubsection*{Part b)}
Let $(\vec{\varphi}, \alpha) \in \mathcal{D}(A)$ and $j \in \{1, 2, \ldots, N \}$. Then, 
\begin{align}
\norm{\Delta J_{j}[(\vec{\varphi}, \alpha)]}_{L^2(\mathbb{R}^3)}
&=  \| \Delta (N^{1/3} \vphc^{\alpha} - 1) \varphi_j \|_2
\leq \norm{\Delta \varphi_j}_2 + N^{1/3} \norm{\Delta \vphc^{\alpha} \varphi_j}_2.
\end{align}
To bound the second summand, we note that $\| \tilde{\eta} \|_{L_2^2(\mathbb{R}^3)}  \leq (1 + \Lambda^2 ) \| \tilde{\eta} \|_2 \leq  \Lambda^3$ and estimate
\begin{equation}\label{eq: solution theory relation for the field operators}
\begin{split}
\norm{\vphc^{\alpha}}_{\infty}  &\leq 2 \int d^3k \, \abs{\tilde{\eta}(k)} \, \abs{\alpha(k)} 
\leq 2 \norm{\tilde{\eta}}_2 \norm{\alpha}_2
\leq 2 \Lambda^3  \norm{\alpha}_2, \\
\norm{\nabla \vphc^{\alpha}}_{\infty} &\leq 2 \int d^3k \, \abs{k} \, \abs{\tilde{\eta}(k)} \, \abs{\alpha(k)}
\leq 2 \norm{ \abs{\cdot} \tilde{\eta}}_2 \norm{\alpha}_2
\leq 2 \Lambda^3  \norm{\alpha}_2, \\
\norm{\Delta \vphc^{\alpha}}_{\infty}  &\leq 2 \int d^3k \, \abs{k}^2 \, \abs{\tilde{\eta}(k)} \, \abs{\alpha(k)}
\leq 2 \norm{ \abs{\cdot}^2 \tilde{\eta}}_2 \norm{\alpha}_2
\leq 2 \Lambda^3  \norm{\alpha}_2.
\end{split}
\end{equation}
Thus,
\begin{align}
\label{eq: solution theory L-2-norm of nabla pvhc varphi}
\norm{\Delta \vphc^{\alpha} \varphi_j}_2 
&\leq \norm{\vphc^{\alpha}}_{\infty} \norm{ \Delta \varphi_j}_2
+ \norm{\Delta \vphc^{\alpha}}_{\infty} \norm{\varphi_j}_2
+ 2 \norm{\nabla \vphc^{\alpha}}_{\infty} \norm{\nabla \varphi_j}_2 
\nonumber \\
&\leq C \Lambda^3  \norm{\alpha}_2   \norm{\varphi_j}_{H^2(\mathbb{R}^3)}
\end{align}
and
$\norm{\Delta J_{j}[(\vec{\varphi}, \alpha)]}_{L^2(\mathbb{R}^3)}
\leq  \big( 1 
+ C N^{1/3} \Lambda^3  \norm{\alpha}_2  \big) \norm{\varphi_j}_{H^2(\mathbb{R}^3)}$. This shows that $J_{j}[(\vec{\varphi}, \alpha)] \in H^2(\mathbb{R}^3)$ for $j \in \{ 1, 2, \ldots, N \}$.
The last component of $J$  is estimated by
$\norm{J_{N+1}[(\vec{\varphi}, \alpha)]}_{L_1^2(\mathbb{R}^3)}
\leq \norm{\alpha}_{L_1^2(\mathbb{R}^3)}
+ (2 \pi)^{3/2} N^{-1}  \norm{\mathcal{F}[\rho_{\vec{\varphi}}]}_{L^{\infty}(\mathbb{R}^3)}  \norm{\tilde{\eta}}_{L_1^2(\mathbb{R}^3)}$.
So if we use 
\begin{align}
\label{eq: Appendix b- estimate for the infinity norm of rho}
 \norm{\mathcal{F}[\rho_{\vec{\varphi}}]}_{L^{\infty}(\mathbb{R}^3)}
 &\leq (2 \pi)^{-3/2} \norm{\rho_{\vec{\varphi}}}_{L^1(\mathbb{R}^3)}
= (2 \pi)^{-3/2} \sum_{j=1}^N \norm{\varphi_j}_2^2
\end{align}
and $\norm{\tilde{\eta}}_{L_1^2(\mathbb{R}^3)} \leq \Lambda^2$ we get
$\norm{J_{N+1}[(\vec{\varphi}, \alpha)]}_{L_1^2(\mathbb{R}^3)}
\leq \norm{\alpha}_{L_1^2(\mathbb{R}^3)}
+ \Lambda^2 N^{-1}\sum_{j=1}^N \norm{\varphi_j}_2^2$.
Hence, $J_{N+1}[(\vec{\varphi}, \alpha)] \in L_1^2(\mathbb{R}^3)$ and thus $J[(\vec{\varphi}, \alpha)] \in \mathcal{D}(A)$.

\subsubsection*{Part c)}

To show part c) of Lemma \ref{lemma: solution theory FSKGE}, we note that the classical radiation field $\vphc^{\alpha}$ is linear in $\alpha$, i.e.,
\begin{align}
\vphc^{\alpha} + \vphc^{\beta} = \vphc^{\alpha + \beta}
\quad \text{and} \quad
\vphc^{\lambda \alpha} = \lambda \vphc^{\alpha} ~\forall \lambda\in\RRR.
\end{align}
For $j \in \{ 1, \ldots, N \}$ we can write
\begin{align}
\label{eq: solution of the FSKGE difference of the Js}
i \big( J[(\vec{\varphi}, \alpha)] - J[(\vec{\psi}, \beta)]  \big)_j
&=  (N^{1/3} \vphc^{\alpha} -1) \varphi_j - (N^{1/3} \vphc^{\beta} -1) \psi_j
\nonumber \\
&= - (\varphi_j - \psi_j) + N^{1/3} \vphc^{\alpha} ( \varphi_j - \psi_j )  + N^{1/3}  \vphc^{\alpha - \beta}  \psi_j
\end{align}
and estimate
\begin{align}
\norm{\big( J[(\vec{\varphi}, \alpha)] - J[(\vec{\psi}, \beta)]  \big)_j}_{L^2(\mathbb{R}^3)}
&\leq (1 + N^{1/3} \Lambda \norm{\alpha}_2) \norm{\varphi_j - \psi_j}_2 + N^{1/3} \Lambda \norm{\psi_j}_2 \norm{\alpha - \beta}_2
\end{align}
by means of \eqref{eq: solution theory relation for the field operators} and $\norm{\tilde{\eta}} \leq \Lambda$.
Hence
\begin{align}
&\sum_{j=1}^N  \norm{\big( J[(\vec{\varphi}, \alpha)] - J[(\vec{\psi}, \beta)]  \big)_j}_{L^2(\mathbb{R}^3)}^2 \nonumber\\
&\quad\leq  2 (1 + N^{1/3} \Lambda \norm{\alpha}_2)^2   \sum_{j=1}^N  \norm{\varphi_j - \psi_j}_2^{2} + 2 N^{2/3} \Lambda^2 \sum_{j=1}^N  \norm{\psi_j}_2^2 \norm{\alpha - \beta}_2^2
\nonumber \\
&\quad\leq  4 \Big( 1 + N^{2/3} \Lambda^2 \big( \|(\vec{\varphi}, \alpha)\|_{\Hilbert}^2 +  \|(\vec{\psi}, \beta)\|_{\Hilbert}^2 \big) \Big) \|(\vec{\varphi}, \alpha) -  (\vec{\psi}, \beta)\|_{\Hilbert}^2.
\end{align}
In order to estimate the difference of $J[(\vec{\varphi},\alpha)]_{N+1}$ and $J[(\vec{\psi},\beta)]_{N+1}$ we note that 
\begin{align}
\rho_{\vec{\varphi}}(x) - \rho_{\vec{\psi}}(x) 
&= \sum_{j=1}^N \overline{\varphi_j(x)} (\varphi_j(x) - \psi_j(x))
+ \sum_{j=1}^N \psi_j(x) ( \overline{\varphi_j(x)} - \overline{\psi_j(x)}).
\end{align}
Thus,
\begin{align}
\norm{\rho_{\vec{\varphi}} - \rho_{\vec{\psi}}}_{L^1(\mathbb{R}^3)}
&\leq \sum_{j=1}^N \norm{\varphi_j}_2 \norm{\varphi_j - \psi_j}_2
+  \sum_{j=1}^N \norm{\psi_j}_2 \norm{\varphi_j - \psi_j}_2
\nonumber \\
&\leq \left( \|(\vec{\varphi}, \alpha)\|_{\Hilbert} +  \|(\vec{\psi}, \beta)\|_{\Hilbert} \right) \|(\vec{\varphi}, \alpha) -  (\vec{\psi}, \beta)\|_{\Hilbert}.
\end{align}
So if we use the linearity of the Fourier transform we obtain
\begin{align}
&\norm{\big( J[(\vec{\varphi}, \alpha)] - J[(\vec{\psi}, \beta)]  \big)_{N+1}}_{L^2(\mathbb{R}^3)} \nonumber\\
&\quad\leq \norm{\alpha - \beta}_2 + N^{-1} (2 \pi)^{3/2} \norm{\tilde{\eta} \big(\mathcal{F}[\rho_{\vec{\varphi}}]  - \mathcal{F}[\rho_{\vec{\psi}}] \big)}_2
\nonumber \\
&\quad\leq   \norm{\alpha - \beta}_2 
+ N^{-1} (2 \pi)^{3/2} \norm{\tilde{\eta}}_2 \norm{\mathcal{F}[\rho_{\vec{\varphi}}]  - \mathcal{F}[\rho_{\vec{\psi}}] }_{\infty}
\nonumber \\
&\quad\leq \norm{\alpha - \beta}_2 
+ N^{-1} \Lambda \norm{\rho_{\vec{\varphi}} - \rho_{\vec{\psi}}}_{L^1(\mathbb{R}^3)} \nonumber \\
&\quad\leq  \left( 1 +
 N^{-1} \Lambda \left( \|(\vec{\varphi}, \alpha)\|_{\Hilbert} +  \|(\vec{\psi}, \beta)\|_{\Hilbert} \right) \right) \|(\vec{\varphi}, \alpha) -  (\vec{\psi}, \beta)\|_{\Hilbert}.
\end{align}
In total we get
\begin{align}
\label{eq: Appendix solution theory proof lemma part c}
 \norm{J[(\vec{\varphi}, \alpha)] - J[(\vec{\psi}, \beta)] }_{\Hilbert}^2
&\leq  8 \Lambda^2 N^{2/3} \Big( 1 +  \|(\vec{\varphi}, \alpha)\|_{\Hilbert}^2 +  \|(\vec{\psi}, \beta)\|_{\Hilbert}^2  \Big) \|(\vec{\varphi}, \alpha) -  (\vec{\psi}, \beta)\|_{\Hilbert}^2.
\end{align}

\subsubsection*{Part d)}

For $j \in \{1, \ldots, N \}$ we consider
\begin{align}
\norm{(A J[(\vec{\varphi}, \alpha)])_j}_{L^2(\mathbb{R}^3)}
&= \norm{(1 - N^{-1/3} \Delta)(N^{1/3} \vphc^{\alpha} -1 )\varphi_j}_{2}
\nonumber\\
&\leq \norm{(1 - N^{-1/3} \Delta) \varphi_j}_2 + C \Lambda^3 \norm{\alpha}_2 \big( N^{1/3} \norm{\varphi_j}_2 + \norm{\varphi_j}_{H^2(\mathbb{R}^3)}  \big)  ,
\end{align}
where we made use of \eqref{eq: solution theory relation for the field operators} and \eqref{eq: solution theory L-2-norm of nabla pvhc varphi}. 
By means of 
\begin{align}
\begin{split}
\label{eq: solution theory of the FSKGE relation beetween the norms}
\norm{\varphi}_2 &\leq \norm{(1 - N^{-1/3} \Delta) \varphi}_2
, \quad 
  \norm{\Delta \varphi}_2 \leq  N^{1/3} \norm{(1 - N^{-1/3} \Delta) \varphi}_2, \\
 \norm{\varphi}_{H^2(\mathbb{R}^3)} &\leq 2  N^{1/3} \norm{(1 - N^{-1/3} \Delta) \varphi}_2
\end{split}
\end{align}
we get
\begin{align}
&\norm{(A J[(\vec{\varphi}, \alpha)])_j}_{L^2(\mathbb{R}^3)}
\leq (1 + C N^{1/3} \Lambda^3 \norm{\alpha}_2)  \norm{(1 - N^{-1/3} \Delta) \varphi_j}_2.
\end{align}
Similarly, we have
$ \norm{\alpha}_2 \leq \norm{\left( 1 + N^{-1/3} \delta_N \omega \right) \alpha}_2$
and obtain
\begin{align}
\label{eq: solution theory of the FSKGE relation between the norms 2}
\norm{(\vec{\varphi},\alpha)}_{\Hilbert} 
&\leq \bigg( \sum_{j=1}^N \norm{(1 - N^{-1/3} \Delta) \varphi_j}_2^{2} + \norm{(1 + N^{-1/3} \delta_N \omega) \alpha}_2^{2} \bigg)^{1/2} 
= \norm{A(\vec{\varphi},\alpha)}_{\Hilbert}.
\end{align}
With \eqref{eq: Appendix b- estimate for the infinity norm of rho} we estimate
\begin{align}
\norm{(A J[(\vec{\varphi}, \alpha)])_{N+1}}_{L^2(\mathbb{R}^3)} 
& \leq \norm{(1 + N^{-1/3} \delta_N \omega) \alpha}_2
\nonumber \\
&\quad + N^{-1} (2 \pi)^{3/2} \left(\norm{\tilde{\eta}}_2  + N^{-1/3}\delta_N \norm{\omega\tilde{\eta}}_2 \right) \norm{\mathcal{F}[\rho_{\vec{\varphi}}]}_{\infty}
\nonumber \\
&\leq \norm{A_{N+1}\alpha}_2 + 2N^{-1} \Lambda^2(1+\sqrt{m}) \sum_{j=1}^N \norm{\varphi_j}_2^2
\nonumber \\
&\leq \Big( 1 + 2N^{-1} \Lambda^2(1+\sqrt{m})  \Big) \norm{A(\vec{\varphi},\alpha)}_{\Hilbert}.
\end{align}
Altogether, this yields
\begin{align}
\norm{A J[(\vec{\varphi}, \alpha)]}^2_{\Hilbert} 
&\leq C N^{2/3} \Lambda^{6} (1+\sqrt{m})^2\big( 1 +   \norm{(\varphi, \alpha)}_{\Hilbert} \big)^2 \norm{A(\vec{\varphi},\alpha)}_{\Hilbert}^2.
\end{align}

\subsubsection*{Part e)}

Note that
\begin{align}
&\sum_{j=1} ^N \norm{(A J[(\vec{\varphi}, \alpha)] - A J[(\vec{\psi}, \beta)])_j}_{L^2(\mathbb{R}^3)}^2 \nonumber\\
&\quad= \sum_{j=1}^N \norm{(1 - N^{-1/3} \Delta) ( J[(\vec{\varphi}, \alpha)] -  J[(\vec{\psi}, \beta)])_j}_2^2 
\nonumber \\
&\quad\leq  2 \Big( \norm{( J[(\vec{\varphi}, \alpha)] -  J[(\vec{\psi}, \beta)])}_{\Hilbert}^2
+ N^{-2/3} \sum_{j=1}^N  \norm{(- \Delta) ( J[(\vec{\varphi}, \alpha)] -  J[(\vec{\psi}, \beta)])_j}_2^2  \Big).
\end{align}
So if we recall \eqref{eq: solution of the FSKGE difference of the Js} and \eqref{eq: Appendix solution theory proof lemma part c} we get
\begin{align}
&\sum_{j=1} ^N \norm{(A J[(\vec{\varphi}, \alpha)] - A J[(\vec{\psi}, \beta)])_j}_{L^2(\mathbb{R}^3)}^2 \nonumber\\
&\quad\leq  8 N^{2/3} \Lambda^2  \big( 1 +  \big( \|(\vec{\varphi}, \alpha)\|_{\Hilbert}^2 +  \|(\vec{\psi}, \beta)\|_{\Hilbert}^2 \big) \big) \|(\vec{\varphi}, \alpha) -  (\vec{\psi}, \beta)\|_{\Hilbert}^2
\nonumber \\
&\qquad +  N^{-2/3} \sum_{j=1}^N \norm{(- \Delta) (\varphi_j - \psi_j)}_2^2
+  \sum_{j=1}^N \norm{(- \Delta)  \vphc^{\alpha} (\varphi_j - \psi_j) }_2^2
+  \sum_{j=1}^N \norm{(- \Delta)  \vphc^{\alpha - \beta} \psi_j }_2^2 .
\end{align}
By means of \eqref{eq: solution theory L-2-norm of nabla pvhc varphi}, \eqref{eq: solution theory of the FSKGE relation beetween the norms} and \eqref{eq: solution theory of the FSKGE relation between the norms 2}
we get
\begin{equation}
\begin{split}
N^{-2/3} \sum_{j=1}^N \norm{( - \Delta) (\varphi_j - \psi_j)}_2^2 
&\leq \norm{A(\vec{\varphi},\alpha) - A(\vec{\psi},\beta)}_{\Hilbert}^2 ,
\\
 \sum_{j=1}^N \norm{( - \Delta) \vphc^{\alpha} (\varphi_j - \psi_j)}_2^2 
&\leq C N^{2/3} \Lambda^6 \norm{(\vec{\varphi},\alpha)}_{\Hilbert}^2   \norm{A(\vec{\varphi},\alpha) - A(\vec{\psi},\beta)}_{\Hilbert}^2, \\
\sum_{j=1}^N \norm{(- \Delta)  \vphc^{\alpha - \beta} \psi_j }_2^2 
&\leq  C N^{2/3} \Lambda^6 \norm{A(\vec{\psi},\beta)}_{\Hilbert}^2   
\norm{A(\vec{\varphi},\alpha) - A(\vec{\psi},\beta)}_{\Hilbert}^2.
\end{split}
\end{equation}
Thus,
\begin{align}
&\sum_{j=1} ^N \norm{(A J[(\vec{\varphi}, \alpha)] - A J[(\vec{\psi}, \beta)])_j}_{L^2(\mathbb{R}^3)}^2 
\nonumber \\
&\quad\leq   C N^{2/3} \Lambda^6  \big( 1 + \norm{(\vec{\varphi},\alpha)}_{\Hilbert}^2 + \|(\vec{\psi},\beta)\|_{\Hilbert}^2 + \norm{A(\vec{\psi},\beta)}_{\Hilbert}^2  \big) 
 \norm{A(\vec{\varphi},\alpha) - A(\vec{\psi},\beta)}_{\Hilbert}^2.
\end{align}
On the other hand we have
\small
\begin{align}
&\norm{(A J[(\vec{\varphi}, \alpha)] - A J[(\vec{\psi}, \beta)])_{N+1}}_{L^2(\mathbb{R}^3)}^2 \nonumber\\
&\quad\leq 2 \norm{(1 + N^{-1/3} \delta_N \omega) (\alpha - \beta)}_2^2
+ 2 N^{-2} (2 \pi)^3 \norm{(1 + N^{-1/3} \delta_N \omega) \tilde{\eta} \big( \mathcal{F}[\rho_{\vec{\varphi}}] - \mathcal{F}[\rho_{\vec{\psi}}] \big)}_2^2 \nonumber\\
&\quad\leq 2 \norm{A (\vec{\varphi},\alpha) - A(\vec{\psi},\beta)}_{\Hilbert}^2 
\nonumber \\
&\qquad + 2 N^{-1}  \big( \|(\vec{\varphi}, \alpha)\|_{\Hilbert} 
+  \|(\vec{\psi}, \beta)\|_{\Hilbert}   \big)^2 \norm{(1 + N^{-1/3} \delta_N \omega) \tilde{\eta}}_2^2 \norm{(\vec{\varphi}, \alpha) - (\vec{\psi}, \beta)}_{\Hilbert}^2
\nonumber \\
&\quad\leq 2 \Big( 1 
+ N^{-1} \big( \|(\vec{\varphi}, \alpha)\|_{\Hilbert} 
+  \|(\vec{\psi}, \beta)\|_{\Hilbert}   \big)^2
\big( 1 + N^{-1/3} \delta_N \sqrt{\Lambda^2 + m^2} \big)^2 \Lambda^2
\Big)  \norm{A (\vec{\varphi},\alpha) - A(\vec{\psi},\beta)}_{\Hilbert}^2 
\nonumber \\
&\quad\leq  8 \Lambda^4 \Big( 1 
+  \big( \|(\vec{\varphi}, \alpha)\|^2_{\Hilbert} 
+  \|(\vec{\psi}, \beta)\|^2_{\Hilbert}   \big)
\big( 1 + N^{-1/3} \delta_N \sqrt{1 + m^2} \big)^2
\Big)  \norm{A (\vec{\varphi},\alpha) - A(\vec{\psi},\beta)}_{\Hilbert}^2 .
\end{align}
\normalsize
In total, we get
\begin{align}
\norm{A J[(\vec{\varphi}, \alpha)] - A J[(\vec{\psi}, \beta)]}_{\Hilbert}^2
&\leq  C N^{2/3} \Lambda^6 \left( 1 + \norm{(\vec{\varphi},\alpha)}_{\Hilbert}^2 + \|(\vec{\psi},\beta)\|_{\Hilbert}^2 + \norm{A(\vec{\psi},\beta)}_{\Hilbert}^2  \right) \nonumber\\
&\quad \times \big( 1 + N^{-1/3} \delta_N \sqrt{1 + m^2} \big)^2  \norm{A(\vec{\psi},\alpha) - A(\vec{\psi},\beta)}_{\Hilbert}^2.
\end{align}

\subsubsection*{Final Statement of the Lemma}

Let $(\vec{\varphi}^{\,0},\alpha^0) \in \mathcal{D}(A)$ and assume there is a $T >0$ so that \eqref{eq: solution theory effective_eqs} has a unique continuously differentiable solution for $t \in [0,T)$. Then,
\begin{align}
\frac{d}{dt} \norm{\varphi_j^t}_2^2
&=  \frac{d}{dt} \scp{\varphi_j^t}{\varphi_j^t} = 2 \Im \scp{\varphi_j^t}{\big(- N^{-1/3}  \Delta + \vphc^{\alpha^t} \big)\varphi_j^t} =0 
\end{align}
because $\vphc^{\alpha^t} \in \mathbb{R}$. Moreover, we can apply Lemma \ref{lemma: growth of alpha-t} locally and conclude $\norm{\alpha^t}_2 \leq \norm{\alpha^0}_2 + C \Lambda t \leq \norm{\alpha^0}_2 + C \Lambda T$. For all $t \in [0,T)$ this shows 
\begin{align}
\norm{(\vec{\varphi}^{\,t}, \alpha^t)}_{\Hilbert}
&\leq \norm{(\vec{\varphi}^{\,0}, \alpha^0)}_{\Hilbert} 
+ \norm{\alpha^0}_2 + C \Lambda T.
\end{align}
\end{proof}

\section*{Acknowledgments}
We would like to thank Peter Pickl and Robert Seiringer for fruitful discussions, and the anonymous referees for their valuable comments and suggestions. Moreover, we would like to thank Niels Benedikter and L{\'a}szl{\'o} Erd\H{o}s for helpful remarks about the semiclassical structure and the Schr\"odinger-Klein-Gordon equations. N.\,L.\ gratefully acknowledges financial support by the European Research Council (ERC) under the European Union's Horizon 2020 research and innovation programme (grant agreement No 694227) and funding for his stay at Princeton University from the project ``Effective One-Particle Equations for Correlated Many-Particle-(Coulomb) Systems: Derivation and Properties'' (project number 318342445) of the German Research Foundation (DFG). S.\,P.\ gratefully acknowledges support from the German Academic Exchange Service (DAAD) and the National Science Foundation under agreement No.\ DMS-1128155. Moreover, we would like to thank Princeton University and the Institute for Advanced Study for their hospitality. S.\,P.\ would additionally like to thank the University of Washington for hospitality.

{}


\begin{thebibliography}{11}

\addcontentsline{toc}{chapter}{Bibliography}

\bibitem{ammarifalconi}
Z.~Ammari and M.~Falconi,
Bohr's correspondence principle for the renormalized Nelson model.
\emph{SIAM J. Math. Anal.} 49(6), 5031--5095 (2017).

\bibitem{anapolitanosmott}
I.~Anapolitanos and M.~Hott,
A simple proof of convergence to the Hartree dynamics in Sobolev trace norms.
\emph{J. Math. Phys.} 57, 122108 (2016).

\bibitem{bach:2015}
V.~Bach, S.~Breteaux, S.~Petrat, P.~Pickl, and T.~Tzaneteas,
Kinetic energy estimates for the accuracy of the time-dependent Hartree-Fock approximation with Coulomb interaction.
\emph{J. Math. Pures Appl.} 105(1), 1--30 (2016).

\bibitem{bardos:2007}
C.~Bardos, B.~Ducomet, F.~Golse, A.\,D.~Gottlieb, and N.\,J.~Mauser,
The {TDHF} approximation for {H}amiltonians with m-particle interaction potentials.
\emph{Commun. Math. Sci.} 5, 1--9 (2007).

\bibitem{bardos:2003}
C.~Bardos, F.~Golse, A.\,D.~Gottlieb, and N.\,J.~Mauser,
Mean field dynamics of fermions and the time-dependent Hartree-Fock equation.
\emph{J. Math. Pures Appl.} 82(6), 665--683 (2003).

\bibitem{bardos:2004}
C.~Bardos, F.~Golse, A.\,D.~Gottlieb, and N.\,J.~Mauser,
Accuracy of the time-dependent Hartree-Fock approximation for uncorrelated initial states.
\emph{J. Stat. Phys.} 115(3--4), 1037--1055 (2004).

\bibitem{benedikter:2015_2}
N.~Benedikter, V.~Jak{\v{s}}i{\'{c}}, M.~Porta, C.~Saffirio, and B.~Schlein,
Mean-field evolution of fermionic mixed states.
\emph{Commun. Pure Appl. Math.} 69(12) (2016).

\bibitem{benedikter:2015}
N.~Benedikter, M.~Porta, C.~Saffirio, and B.~Schlein,
From the {H}artree dynamics to the {V}lasov equation.
\emph{Arch. Rat. Mech. Anal.} 221(1) (2016).

\bibitem{benedikter:2014}
N.~Benedikter, M.~Porta, and B.~Schlein,
Mean-field dynamics of fermions with relativistic dispersion.
\emph{J. Math. Phys.} 55(2) (2014).

\bibitem{benedikter_2013}
N.~Benedikter, M.~Porta, and B.~Schlein,
Mean-field evolution of fermionic systems.
\emph{Commun. Math. Phys.} 331(3), 1087--1131 (2014).

\bibitem{schleinbook}
N.~{Benedikter}, M.~{Porta}, and B.~{Schlein},
\emph{Effective Evolution Equations from Quantum Dynamics}
(SpringerBriefs in Mathematical Physics, Cambridge, 2016)

\bibitem{benedikter_2018}
N.~Benedikter, J.~Sok, and J.\,P.~Solovej
The Dirac-Frenkel principle for reduced density matrices, and the Bogoliubov-de Gennes equations
\emph{Ann. H. Poincar\'e} 19(4), 1167--1214 (2018).

\bibitem{braunhepp}
W.~Braun and K.~Hepp,
The Vlasov dynamics and its fluctuations in the $1/N$ limit of interacting classical particles.
\emph{Commun. Math. Phys.} 56(2), 101--113 (1977).

\bibitem{colliander_2008}
J.~Colliander, J.~Holmer and N.~Tzirakis,
Low regularity global well-posedness for the {Z}akharov and {K}lein-{G}ordon-{S}chr\"odinger systems.
\emph{Trans. Amer. Math. Soc.} 360(9), 4619--4638 (2008).

\bibitem{correggi_2017}
M.~Correggi and M.~Falconi,
Effective potentials generated by field interaction in the quasi-classical limit.
\emph{Ann. H. Poincar\'e} 19(1), 189--235 (2018).

\bibitem{correggi_2018}
M.~Correggi, M.~Falconi, and M.~Olivieri,
Magnetic Schr\"odinger operators as the quasi-classical limit of Pauli-Fierz-type models.
\emph{J. Spectr. Theor., in press}, \href{https://arxiv.org/pdf/1711.07413v2}{arXiv:1711.07413v2} (2017).

\bibitem{davies_1979}
E.\,B.~Davies,
Particle-boson interactions and the weak coupling limit.
\emph{J. Math. Phys.} 20, 345--351 (1979).

\bibitem{erdoes:2004}
A.~Elgart, L.~Erd\H{o}s, B.~Schlein, and H.-T.~Yau,
Nonlinear {H}artree equation as the mean field limit of weakly coupled fermions.
\emph{J. Math. Pures Appl.} 83(10), 1241--1273 (2004).

\bibitem{falconi2}
M.~Falconi,
Classical limit of the Nelson model,
\url{http://amsdottorato.unibo.it/4631/1/marco_falconi_tesi.pdf}.
Ph.D.\ thesis Universit\`{a} di Bologna (2012).

\bibitem{falconi}
M.~Falconi,
Classical limit of the Nelson model with cutoff.
\emph{J. Math. Phys.} 54(1), 012303 (2013).

\bibitem{frankenschlein}
R.\,L.~Frank and B.~Schlein,
Dynamics of a strongly coupled polaron.
\emph{Lett. Math. Phys.} 104, 911--929 (2014).

\bibitem{frankgang}
R.\,L.~Frank and Z.~Gang,
Derivation of an effective evolution equation for a strongly coupled polaron.
\emph{Anal. PDE} 10(2), 379--422 (2017).

\bibitem{froehlich:2011}
J.~Fr{\"o}hlich and A.~Knowles,
A microscopic derivation of the time-dependent {H}artree-{F}ock equation with {C}oulomb two-body interaction.
\emph{J. Stat. Phys.} 145(1), 23--50 (2011).

\bibitem{ginibrenironivelo}
J.~Ginibre, F.~Nironi, and G.~Velo,
Partially classical limit of the Nelson model.
\emph{Ann. H. Poincar\'e} 7, 21--43 (2006).

\bibitem{ginibrevelo}
J.~Ginibre and G.~Velo,
The classical field limit of scattering theory for nonrelativistic many-boson systems I and II.
\emph{Commun. Math. Phys.} 66(1), 37--76 (1979) and 68(1), 45--68 (1979).

\bibitem{griesemer}
M.~Griesemer,
On the dynamics of polarons in the strong-coupling limit.
\emph{Rev. Math. Phys.} 29(10), 1750030 (2017).
 
\bibitem{hepp}
K.~Hepp, 
The classical limit for quantum mechanical correlation functions.
\emph{Commun. Math. Phys.} 35, 265--277 (1974). 

\bibitem{hiroshima_1998}
F.~Hiroshima,
Weak coupling limit with a removal of an ultraviolet cutoff for a Hamiltonian of particles interacting with a massive scalar field.
\emph{Infin. Dimens. Anal. Qu.} 1, 407--423 (1998).

\bibitem{lanford} 
O.\,E.~Lanford,
Time evolution of large classical systems.
In: J. Moser (ed.) \emph{Dynamical Systems, theory and applications}, volume 38 of Lecture Notes in Physics, 1--111, Springer (1975).

\bibitem{leopold2}
N.~Leopold,
Effective evolution equations from quantum mechanics, \url{https://edoc.ub.uni-muenchen.de/21926/}. Ph.D. thesis LMU M\"unchen (2018).

\bibitem{leopold}
N.~Leopold and P.~Pickl,
Derivation of the Maxwell-Schr\"odinger equations from the Pauli-Fierz Hamiltonian.
\emph{Preprint}, \href{https://arxiv.org/pdf/1609.01545v2}{arXiv:1609.01545v2} (2016).  

\bibitem{leopold3}
N.~Leopold and P.~Pickl,
Mean-field limits of particles in interaction with quantized radiation fields.
In: D.~Cadamuro, M.~Duell, W.~Dybalski, and S.~Simonella (eds) \emph{Macroscopic Limits of Quantum Systems}, volume 270 of Springer Proceedings in Mathematics \& Statistics, 185--214 (2018).

\bibitem{lieb:2005}
E.\,H.~Lieb, R.~Seiringer, J.\,P.~Solovej, and J.~Yngvason,
\emph{The Mathematics of the {B}ose Gas and its Condensation}
(Birkh\"auser Basel, Oberwolfach Seminars 34, 2005).

\bibitem{narnhofer:1981}
H.~Narnhofer and G.~L. Sewell,
Vlasov hydrodynamics of a quantum mechanical model.
\emph{Commun. Math. Phys.} 79(1), 9--24 (1981).

\bibitem{nelson}
E.~Nelson, 
Interaction of nonrelativistic particles with a quantized scalar field.
\emph{J. Math. Phys.} 5(9), 1190--1197 (1964).

\bibitem{pecher}
H.~Pecher,
Some new well-posedness results for the {K}lein-{G}ordon-{S}chr\"odinger system.
\emph{Diff. Int. Equations} 25(1/2), 117--142 (2012).

\bibitem{petrat_2015}
S.~Petrat and P.~Pickl,
A new method and a new scaling for deriving fermionic mean-field dynamics.
\emph{Math. Phys. Anal. Geom.} 19(1) (2016).

\bibitem{petrat:2016}
S.~Petrat,
Hartree corrections in a mean-field limit for fermions with {C}oulomb interaction.
\emph{J. Phys. A: Math. Theor.} 50(24) (2017).

\bibitem{pickl1}
P.~Pickl,
A simple derivation of mean field limits for quantum systems.
\emph{Lett. Math. Phys.} 97, 151--164 (2011).

\bibitem{porta:2016}
M.~Porta, S.~Rademacher, C.~Saffirio, and B.~Schlein,
Mean field evolution of fermions with {C}oulomb interaction.
\emph{J. Stat. Phys.} 166(6), 1345--1364 (2017).

\bibitem{reedsimon}
M.~Reed and B.~Simon,
\emph{Methods of Modern Mathematical Physics II: Fourier Analysis, Self-Adjointness}
(Academic Press, Inc., San Diego, 1975).

\bibitem{rodnianskischlein}
I.~Rodnianski and B.~Schlein,
Quantum fluctuations and rate of convergence towards mean field dynamics. \emph{Commun. Math. Phys.} 291(1), 31--61 (2009).

\bibitem{saffirio_proc}
C.~Saffirio,
Mean-field evolution of fermions with singular interaction. In: D.~Cadamuro, M.~Duell, W.~Dybalski, and S.~Simonella (eds) \emph{Macroscopic Limits of Quantum Systems}, volume 270 of Springer Proceedings in Mathematics \& Statistics, 81--99 (2018).

\bibitem{mcgriesemer}
J.~Schmid and M.~Griesemer,
Well-posedness of non-autonomous linear evolution equations in uniformly convex spaces.
\emph{Math. Nachr.} 290(2--3), 435--441 (2017).

\bibitem{spohneffbook}
H.~Spohn,
\emph{Large Scale Dynamics of Interacting Particles}
(Springer-Verlag, Berlin Heidelberg, ed.~1, 1991).

\bibitem{spohn:1981}
H.~Spohn,
On the {V}lasov hierarchy.
\emph{Math. Methods Appl. Sci.} 3(1), 445--455 (1981).

\bibitem{spohn}
H.~Spohn,
\emph{Dynamics of charged particles and their radiation field}
(Cambridge University Press, Cambridge, 2004).
 
\bibitem{teschl}
G.~Teschl,
\emph{Mathematical Methods in Quantum Mechanics: With Applications to Schr\"odinger Operators}
(American Mathematical Society, Graduate Studies in Mathematics Volume 157,  Providence, Rhode Island, 2014).

\bibitem{teufel2}
S.~Teufel, 
Effective $N$-body dynamics for the massless Nelson model and adiabatic decoupling without spectral gap.
\emph{Ann. H. Poincar\'e} 3, 939--965 (2002).
\end{thebibliography}
\end{document}